\documentclass[onecolumn,prx, superscriptaddress,longbibliography,margin=1mm 11pt]{revtex4-2} 

\usepackage{algorithm}
\usepackage{ragged2e}
\usepackage{graphicx}
\usepackage{bm}
\usepackage{amsmath}
\usepackage{environ}
\usepackage{appendix}
\usepackage{physics}
\usepackage{rotating,booktabs,array,amsmath,amssymb}

\usepackage{algorithmicx,algorithm,setspace,subfigure}
\usepackage{float}

\usepackage{enumitem}

\newcolumntype{C}{>{$\displaystyle}c<{$}} % automatic display-style math mode
\usepackage{caption}
\usepackage{subcaption}

\usepackage[nopar]{lipsum}

\usepackage{enumitem}
\setlist[enumerate,1]{label=(\arabic*),ref=\arabic*}
\makeatletter
\usepackage{amsthm}
\usepackage{amssymb}
\usepackage{amsfonts}
\usepackage{fancyhdr}
\usepackage{slashed}
\usepackage{bbm}
\usepackage{latexsym,epsfig,bbm}
\usepackage[colorlinks=true,urlcolor=blue,citecolor=blue]{hyperref}
\usepackage{todonotes}
\definecolor{blue}{rgb}{0,0.2,1}

\definecolor{red}{rgb}{0.9,0,0}

\newtheorem{theorem}{Theorem}
\newtheorem{lemma}[theorem]{Lemma}

\newtheorem{corollary}[theorem]{Corollary}
\newtheorem{remark}[theorem]{Remark}
\newtheorem{assumption}[theorem]{Assumption}

\renewcommand{\mathbf}[1]{\boldsymbol{#1}}

\begin{document}

\title{Quantum algorithms for  viscosity solutions to nonlinear Hamilton-Jacobi equations based on an entropy penalisation method}

\author{Shi Jin}

\affiliation{Institute of Natural Sciences, Shanghai Jiao Tong University, Shanghai 200240, China}
\affiliation{School of Mathematical Sciences, Shanghai Jiao Tong University, Shanghai, 200240, China}
\affiliation{Ministry of Education Key Laboratory in Scientific and Engineering Computing, Shanghai Jiao Tong University, Shanghai 200240, China}
\author{Nana Liu}
\email{nana.liu@quantumlah.org}
\affiliation{Institute of Natural Sciences, Shanghai Jiao Tong University, Shanghai 200240, China}
\affiliation{School of Mathematical Sciences, Shanghai Jiao Tong University, Shanghai, 200240, China}
\affiliation{Ministry of Education Key Laboratory in Scientific and Engineering Computing, Shanghai Jiao Tong University, Shanghai 200240, China}
\affiliation{Global College, Shanghai Jiao Tong University, Shanghai 200240, China.}

\date{\today}

\begin{abstract}
We present a framework for efficient extraction of the viscosity solutions of nonlinear Hamilton–Jacobi equations with convex Hamiltonians. These viscosity solutions play a central role in areas such as front propagation, mean-field games, optimal control, machine learning, and a direct application to the forced Burgers' equation. Our method is based  on an  entropy penalisation method proposed by Gomes and Valdinoci \cite{gomes2007entropy}, which generalises the Cole–Hopf transform from quadratic to general convex Hamiltonians, allowing an approximation of viscous Hamilton-Jacobi dynamics by a discrete-time linear dynamics which approximates a linear heat-like parabolic equation, and can also extend to continuous-time dynamics. This makes the method suitable for quantum simulation.  The validity of these results hold for arbitrary  nonlinearity that correspond to convex Hamiltonians, and for arbitrarily long times, thus obviating a chief obstacle in most quantum algorithms for nonlinear partial differential equations. We provide quantum algorithms -- both analog and digital -- for extracting pointwise values, gradients, minima, and function evaluations at the minimiser of the viscosity solution, without requiring nonlinear updates or full state reconstruction. 
\end{abstract}

\maketitle

\tableofcontents 

\section{Introduction}

The simulation of partial differential equations (PDEs) on quantum computers has emerged as a promising direction in computational mathematics and quantum information science. Most existing quantum algorithms for PDEs, however, are restricted to linear systems. Quantum simulation of nonlinear PDEs remains a formidable  challenge. In this paper we develop efficient quantum algorithms for {\it viscosity solutions} to nonlinear PDEs of Hamilton-Jacobi type,  which appear ubiquitously in mechanics, control theory, and statistical physics.  \\

There are two challenges in developing quantum algorithms for Hamilton-Jacobi equations. The most obvious is the nonlinearity in its Hamiltonian $\mathcal{H}$, which can give rise to  finite time singularity -- caustics -- even if the initial data are smooth. Beyond the singularity, it is necessary to define physically relevant solutions. The first such solution is the multivalued solution, which is relevant to  application domains in geometric optics, seismic waves, semi-classical limit of quantum mechanics and high frequency limits of linear waves. Previously, the authors developed a quantum algorithm for nonlinear Hamilton–Jacobi equations capable of capturing \textit{multi-valued solutions}. The idea there is to use the level set method \cite{Jin-Osher}, that transfers, exactly, the nonlinear  Hamilton-Jacobi equation to the \textit{linear} Liouville equation defined in the phase space \cite{JinLiu-nonlinear}. \\

The second interesting kind of solution beyond caustics are the \textit{viscosity solutions}, introduced by Crandall and Lions \cite{CL1}, which plays a central role in areas such as front propagation, mean-field games,  optimal control and machine learning \cite{E-Jentzen}. In some of these problems,  classical computational methods offer suffer from the curse-of-dimensionality, thus calls for quantum computation. \\

So far, there have been two broad classes of quantum algorithms for nonlinear PDEs. The first one uses {\it analytic} transformations to make the equations linear, which we  call linear representation \cite{jin2022timemultiscale}. The  level set \cite{JinLiu-nonlinear} and von-Neumann Koopman approach \cite{joseph2020koopman} belong to this category. Such methods have no restrictions for the strength of nonlinearity but are limited in its applicability to just few nonlinear  PDEs. The other class of quantum algorithms for nonlinear ODEs and PDEs, for example the Carleman truncation method \cite{Liu-PNAS, Liu-CMS},  approximates the original problems by linear ones, which we classify as linear approximations in \cite{JinLiu-nonlinear}. For nonlinear PDEs like  Hamilton-Jacobi or fluid dynamic equations, such linear approximations are unable to capture essential nonlinear physical phenomena such as caustics, shocks and turbulence, since they appear in physical systems with strong nonlinearity and weak dissipation.  In fact a main restriction of the linear approximation methods is that they are usually valid for weak nonlinearity and with strong dissipations.  The nonlinear Hamilton-Jacobi equation, in the contrary, has strong nonlinearity and weak dissipation--if one is interested in the viscosity solution.\\

 In this paper, we use what we call a \textit{linear formulation}, where we can still approximate the physically relevant solutions -- the viscosity solutions -- of nonlinear Hamilton-Jacobi equations  using a linear representation where the results still remain valid {\it globally} in time and truly nonlinear features can be captured. Usually when one computes  problems with possible singular solutions--such as shocks in gas dynamics or caustics in Hamilton-Jacobi equations--one needs to  use numerical viscosity in order to eliminate numerical oscillations and guarantee convergence to the physically relevant viscosity solutions. Even for a linear $\mathcal{H}$, the numerical viscosity in modern shock or caustic capturing schemes is nonlinear \cite{LeVequeBook}. This adds another difficulty in numerical computation of Hamilton-Jacobi equations.  To avoid this numerical difficulty, we use a {\it linear, artificial} viscosity approximation, which perturbs the Hamilton-Jacobi equation into a viscous Hamilton-Jacobi equation with a linear viscosity term. The concept of artificial viscosity was first introduced by von Neumann and Ritchmyer \cite{vNR} in the 1950's for the computation of shock waves in compressible gas dynamics, and has become a standard numerical tool to handle singular solutions. 
 
 To overcome the problem of nonlinear $\mathcal{H}$, we first observe that for quadratic Hamiltonians, the classical Cole–Hopf transformation converts the viscous Hamilton-Jacobi equation into a linear heat equation, providing an exact mapping -- thus a true linear representation -- between nonlinear and linear dynamics. For more general convex Hamiltonians, we adopt and generalise the entropy penalisation method of Gomes and Valdinoci \cite{gomes2007entropy}, which constructs an entropy-regularised approximation that correspond to the viscous Hamilton-Jacobi equations. Furthermore,  via the Cole-Hopf transform, it gives a {\it linear} discrete-time model that corresponds to a linear parabolic PDE, which is convenient for quantum simulation. This transformation enables a linear formulation for the nonlinear viscous Hamilton-Jacobi equation without loss of essential nonlinear behavior.  A related method was introduced in \cite{Mather} to compute the minimal Mather measure.\\

A second main challenge is that even with either a linear formulation or linear representation, one still needs to find efficient algorithms to extract physical quantities of interest.  In the present work we also develop quantum protocols -- both analog and digital -- to estimate physically relevant quantities. Examples include the solution at any point, the global minimum, the gradient of the solution  at a given point, and the value of an arbitrary function at the minimum point of the solution. \\
 
  In Section~\ref{sec:background} we summarise the classical algorithms to find solutions of the viscous Hamilton-Jacobi equation via linear formulation, along with error bounds to the corresponding viscosity solutions. Here the solutions can be approximated by solving linear parabolic PDEs. In Section~\ref{sec:quantumsimulation} we present the quantum simulation algorithms -- both analog and digital -- for those linear parabolic PDEs. In Section~\ref{sec:observables}, using the quantum states constructed in Section~\ref{sec:quantumsimulation}, we present methods for estimating four types of quantities: (a) solution at a point (b) gradient of the solution at a point (c) global minimum of solution and (d) value of known functions at the minimum point. 

\section{Linear formulation of Hamilton-Jacobi equations} \label{sec:background}

\subsection{Background on viscosity solutions of Hamilton-Jacobi equations}

The Hamilton-Jacobi equation is a classical PDE that takes the general form
\begin{align}\label{hj0}
    \frac{\partial S}{\partial t}+\mathcal{H}(t,x,\nabla S)=0, \qquad S_0(x)=S(0,x),
\end{align}
where $\mathcal{H}(t,x,\nabla S)$ is the Hamiltonian, assumed to be convex and Lipschitz continuous in the canonical momentum variable $\nabla S$. The  space variable is $x\in \mathbb{R}^d$ and time $t>0$. The Hamilton-Jacobi equation has wide range of applications, from classical mechanics \cite{goldstein}, geometric optics \cite{evans2022partial}, to optimal control \cite{Bardi-book} and machine learning \cite{E-Jentzen}. It also arises as a numerical tool -- known as the level set method for computing propagating interfaces \cite{OsherSethian}. For these nonlinear PDEs, even with smooth initial data, the solutions can develop singularities known as caustics (e.g. discontinuity in the gradient) in finite time, which  leads to a breakdown of classical (continuously differentiable) solutions. 

The viscosity solution, introduced by Crandall and Lions \cite{CL1}, is a standard mathematical notion to define a unique weak solution beyond the singularity.  Numerical approximations to Hamilton-Jacobi equations have been extensively studied, see for example \cite{OsherSethian}.  To avoid numerical oscillations due to the presence of caustics, high order numerical schemes need to use slope limiters around caustics -- which reduce the numerical accuracy to first order --  so a sufficient amount of numerical viscosity is used to suppress numerical oscillations and to guarantee viscosity solutions can be obtained \cite{LeVequeBook}.  It's important to remark that these numerical viscosity terms are {\it nonlinear} because they depend on the solutions themselves, as their role is to locally suppress oscillations near regions of high oscillations and in smooth regions they are not used. This nonlinearity is present even when the original equations are linear. Since our goal is to identify a linear formulation for the viscosity solution of Hamilton-Jacobi equations, we do not use these more robust but nonlinear modern shock-capturing schemes (which uses slope limiters)  \cite{OsherSethian}. 

Instead, we use a {\it linear, artificial} viscosity approximation. The simplest such approximation is the following regularised problem, called the viscous Hamilton-Jacobi equation
\begin{align} \label{eq:hjviscosity0}
     \frac{\partial S_{\nu}}{\partial t}+\mathcal{H}(t,x,\nabla S_{\nu})=2 a\nu\sum_{j=1}^d \frac{\partial^2 S_{\nu}}{\partial x^2_j}, \quad S_{\nu}(0,x)=S_{0}(x), \qquad \nu>0,
\end{align}
where the right hand side is the linear, artificial viscosity term with a small viscosity coefficient $0<\nu<<1$., and $a$ is a constant depending on $\mathcal{H}$ (to be given in Lemma \ref{prop-24}).   The concept of artificial viscosity was first introduced by von Neumann and Ritchmyer \cite{vNR} in the 50's for the computation of shock waves in compressible gas dynamics, and has become a standard numerical tool to handle singular solutions, although more modern high-resolution schemes use the more robust nonlinear artificial viscosities as mentioned earlier.  The viscosity term  smooths out the caustics and the solutions of Eq.~\eqref{eq:hjviscosity0}, and in the vanishing viscosity limit $\nu \to 0$, recovers the   viscosity solution defined by Crandall and Lions \cite{CL2}.  This added artificial viscosity introduces an error, the rate of convergence of it has been under extensive mathematical studies, see for examples  \cite{Fleming, CL2, Souga,  Evans-10, lin20011, Brenier, CG25, tran2021hamilton}. The rate of convergence is between  $O(\nu^{1/2})$ and $O(\nu)$, depending on the regularity and convexity of the Hamiltonian and the initial data.
%Since quantum algorithm is less sensitive to the algebraic order of numerical convergence due to the use of qubits, we just give a  result below just to show that the viscosity approximation indeed can give good approximation. Here we assume that $\mathcal{H}$ is convex and  twice-differentiable, and $S_0$ is twice differentiable, then:

%\begin{lemma} \label{lem:viscosityerror}
%    Assume $\mathcal{H}$ is $C^2((0,T)\times \mathbb{R}^n \times \mathbb{R}^n)$, $S_0(x)$ is Lipschitz continuous, and $S$ is the viscosity solution to  Eq.~\eqref{hj0}, then 
%    \begin{align}\label{solution-error}
%        \sup_{0 \leq t \leq T} \sup_{x \in \mathbb{R}^d}|S(t,x)-S_{\nu}(t, x)|\leq c \nu,
%    \end{align}
%      where $c$ depends on $d$, $S_0(x)$, $\mathcal{H}$, and $T$.       
%\end{lemma}
The dependence of the error on  $d$ also varies depending on the nature of $H$ and space dimension. When $\mathcal{H}$ is just Lipschitz continuous, $c$ can be even independent of $d$ while the rate of error is $O(\nu^{1/2})$ \cite{CL2}. This paper studies general Hamiltonian so we will not get into specific situations, instead  we will loosely use the error in Lemma \ref{lem:viscosityerror}, which also gives error estimate for derivatives of the solution. To estimate the approximation of the derivative, we need semi-concavity on $S_0$, and  stronger assumptions on $\mathcal{H}$. Let $Q_T=(0,T)\times \mathbb{T}^d$ where $\mathbb{T}^d$ is a $d$-dimensional torus. Assume $\mathcal{H}(t,x,p)\in C^2((0,T)\times \mathbb{R}^n\times\mathbb{R}^n)$, with
\begin{align}
\sum_{(t,x)\in Q_T}|\nabla_x \mathcal{H}|, \quad 
\sum_{(t,x)\in Q_T}|\nabla^2_{xp} \mathcal{H}| , \quad
\sum_{(t,x)\in Q_T}|\nabla^2_{xx} \mathcal{H} |\le C_H(1+|p|) \nonumber  
\end{align}
and
\begin{align}
\xi^T \, (\nabla^2_{pp} \mathcal{H} )\xi\ge a(t,x) |\xi|^2,  \quad a>0. \nonumber 
\end{align}
Then by combining Theorem 3.1 and Corollary 3.4 from \cite{CG25} gives:

\begin{lemma} \label{lem:viscosityerror}
 For   the viscosity solution $S$ of Eq.~\eqref{hj0}, 
    \begin{align}\label{deriv-error}
     & \sup_{0 \leq t \leq T} \sup_{x\in \mathbb{T}^d} |S(t,x)- S_{\nu}(t, x)|\leq C {\nu}^\beta, \\
 & \sup_{0 \leq t \leq T} \|\nabla S(t,x)-\nabla S_{\nu}(t, x)\|_{l_2(\mathbb{T}^d)}\leq C {\nu}^{\beta/2}, 
    \end{align}
      for all $\beta\in (1/2, 1)$, $C$ is linear in $d$, scales with $T^{3/2}$, and also depends on $\beta,  S_0$ and $H$.
\end{lemma}
%We note that, for our later applications in recovering the observables of $S(t,x)$, we do not require any further bounds with respect to any other norm. \\

\subsection{Linear formulations of viscous nonlinear Hamilton-Jacobi equations} \label{sec:heatquantumsimulation}

To simulate the solutions of nonlinear PDEs for with general nonlinearity, with long time validity, using quantum simulation, one of the important components is being able to find a linear formulation of the nonlinear PDEs. Now $S_{\nu}$ obeying Eq.~\eqref{eq:hjviscosity0} still satisfies a nonlinear PDE and the next step is to find a linear formulation of the viscous Hamilton-Jacobi equation. We will see in Section~\ref{sec:HJmechanics} that for the special Hamiltonian that is quadratic in $\nabla S_{\nu}$, 
one can use the  classical Cole-Hopf transform
\cite{evans2022partial} to transform the nonlinear viscous Hamilton-Jacobi PDE \eqref{eq:hjviscosity0} into a linear heat equation. 

However, for more general convex Hamiltonians, the direct application of the  Cole-Hopf transformation will not lead to a linear heat or parabolic equation. We will see in  Section~\ref{sec:HJgeneral} how the entropy penalisation method introduced by Gomes and Valdinoci in \cite{gomes2007entropy} can be used to provide such a linear formulation. Through a discrete-time process, this method generalises the Cole-Hopf transformation between a general Hamilton-Jacobi equation with convex Hamiltonian and a {\it linear},  discrete-time evolution process that is asymptotically close to a linear parabolic equation when the time step is small.

\subsubsection{Quadratic Hamiltonian} \label{sec:HJmechanics}

We begin with a special example of the quadratically nonlinear viscous Hamilton-Jacobi PDE 
\begin{align} \label{eq:hjviscosity1}
     \frac{\partial S_{\nu}}{\partial t}+\mathcal{H}(t,x,\nabla S_{\nu})=\nu\sum_{j=1}^d \frac{\partial^2 S_{\nu}}{\partial x^2_j}, \qquad \mathcal{H}(t,x,\nabla S_{\nu})=\sum_{j=1}^d \frac{1}{2}\left(\frac{\partial S_{\nu}}{\partial x_j}\right)^2+V(t,x), \quad S_{\nu}(0,x)=S_{0}(x).
\end{align}
% For small $\nu>0$, its solution $S_{\nu}$ is the viscosity approximation of the solution to the original Hamilton-Jacobi equation for $S(t, x)$ with a Newtonian Hamiltonian ${H}(\nabla S, t, x)$:
%\begin{align}\label{inviscidH-J}
%    \frac{\partial S}{\partial t}+{H}(\nabla S, t, x)=0, \qquad S_0(x)=S(0,x)=S_{\nu}(0,x).
%\end{align}
%For example, this class of Hamiltonians $h$ appears in classical mechanics, geometric optics, semi-classical quantum mechanics, optimal transport and certain cases of optimal control (uses a terminal instead of an initial condition).\\ 
Eq.~\eqref{eq:hjviscosity1} can be transformed into a linear PDE (heat equation) using the Cole-Hopf transformation
\cite{evans2022partial}, except here we include a normalisation term $\mathcal{N}_0$ 
\begin{align} \label{eq:ch1}
    S_{\nu}(t,x)=-2 \nu \ln (\mathcal{N}_0u(t, x)), \qquad \mathcal{N}^2_0=\int e^{-S_{\nu}(0,x)/\nu}dx.
\end{align} 
Then it is simple to verify that  $u(t,x)$ satisfies the $d$-dimensional heat equation with a linear source term:
\begin{align} \label{eq:heat1}
    \frac{\partial u}{\partial t}-\frac{V(t, x)}{2\nu}u=\nu\sum_{j=1}^d \frac{\partial^2 u}{\partial x_j^2}, \qquad u_0(x)=u(0,x), \qquad \int |u_0(x)|^2 dx=1,
\end{align}
where the initial state is automatically normalised from our definition of the Cole-Hopf transformation
\begin{align}
    u_0(x)=\frac{e^{-S_{\nu}(0,x)/(2\nu)}}{\mathcal{N}_0} \implies \int |u_0(x)|^2 dx=1,
\end{align}
which is more convenient for embedding into a normalised initial quantum state. The linear PDE in Eq.~\eqref{eq:heat1} can then be simulated on a quantum device - both digital and analog with qubit, qudit or continuous-variable quantum systems. We describe both the analog and digital quantum simulation algorithms for $u(t)$ in the following Section~\ref{sec:quantumsimulation}. \\

For the purpose of simplifying notation in  our algorithm, we can set $\mathcal{N}_0=1$ in the rest of the paper without losing generality, since a different value of $\mathcal{N}_0 \neq 1$ can always be added classically at the end, adding $-2\nu \ln \mathcal{N}_0$ to $S_{\nu}$. For any derivatives of $S_{\nu}$, the value of $\mathcal{N}_0$ does not contribute. \\

\begin{remark}\label{CH-Burgers}
Defining $R_\nu=\nabla S_\nu$. If one takes  gradient on the Hamilton-Jacobi equation \eqref{eq:hjviscosity1}, one gets the forced Burgers' equation
\begin{align}\label{burgers}
     \frac{\partial R_{\nu}}{\partial t}+ R_\nu \cdot \nabla R_\nu + \nabla V =\nu \sum_{j=1}^d \frac{\partial^2 R_{\nu}}{\partial x^2_j}.
\end{align}
Clearly, all of the quantum algorithms developed in this article can also be directly applied to solve this multi-dimensional forced Burgers' equation, when $R_{\nu}$ is curl-free.
\end{remark}

\subsubsection{More general convex Hamiltonians} \label{sec:HJgeneral}
In the previous section, the Hamiltonian is quadratic in the canonical momentum, which is represented by the gradient of $S_{\nu}(t, x)$, and intuitively this term acts like a kinetic energy term. For Hamiltonians with terms that are more general in canonical momentum, however, it is not possible to directly transform the nonlinear PDE into a linear one. So we proceed using an alternative method introduced in \cite{gomes2007entropy}, which still makes use of the Cole-Hopf transformation, though doing so indirectly.

We first define $\mathcal{L}(x,v)$ as a suitably smooth Lagrangian and  $v\in \mathbb{R}^d$ is the velocity. This Lagrangian is conjugate to the Hamiltonian $\mathcal{H}(x,p)$ in the Hamilton-Jacobi equation, with conjugate momentum $p\in \mathbb{R}^d$. This conjugacy can be expressed using the Legendre–Fenchel definition 
\begin{align}
    & \mathcal{L}(x,v)=\sup_p \{p \cdot v-\mathcal{H}(x, p)\}, \\
    & \mathcal{H}(x,p)=\sup_v \{p \cdot v-\mathcal{L}(x, v)\}=-\inf_v \{\mathcal{L}(x,v)-p \cdot v\}.
\end{align}
We assume that this Lagrangian can be written in the form $\mathcal{L}(x,v)=K(v)-V(x)$, where $K(v)$ is strictly convex in $v$ and superlinear at infinity, and $V(x)$ is semiconvex and bounded.

The basic idea proposed in \cite{gomes2007entropy} is that, instead of working directly with the PDE in continuous time and then perform the Cole-Hopf transformation, one first chooses a small time step $h>0$ and proceeds with discrete-time schemes. Here we define at time-step $n$ the solution $S_D^{n}(x)$ via a time-marching process:
       \begin{align} \label{eq:sndef}
             S_D^{n+1}(x)=G[S_D^n](x)=-2\nu \ln \left(\int e^{-(h\mathcal{L}(x,v)+S_D^n(x+hv))/(2\nu)}dv\right),
        \end{align}
where $G$ is a {\it nonlinear} operator that approximately corresponds to the viscous Hamilton-Jacobi equation Eq.~\eqref{eq:hjviscosity0}. The solution of the discrete-time scheme $S_D^n$ can be related to the solution of another discrete-time scheme $\tilde{u}^n$, defined by 
 \begin{align}\label{linear-scheme}
        \tilde{u}^{n+1}(x)= \tilde{L}[\tilde{u}^n](x)=\int  e^{-h\mathcal{L}(x,v)/(2\nu)}\tilde{u}^n(x+hv)dv, 
    \end{align}
Here $\tilde{L}$ is a {\it linear} operator. Then it has been shown that \cite{gomes2007entropy}
$S_D^n(x)$ and $\tilde{u}^n$ are related by  the Cole-Hopf transformation:
\begin{equation}\label{Cole-Hopf}
        S_D^{n}(x)=-2\nu \ln \tilde{u}^n(x).
 \end{equation}
Numerically it will be more convenient to use the normalised 
equation
\begin{align}\label{N-linear-scheme}
        {u}^{n+1}(x)= {L}[{u}^n](x)=\int P(v) {u}^n(x+hv)dv,\qquad P(v) = \frac{e^{-h{K}(v)/(2\nu)}}{\int e^{-h{K}(v)/(2\nu)}dv},
\end{align}
that  satisfies the property $\|L[u]\|\le \|u\|$, hence the dynamics \eqref{N-linear-scheme} is {\it contractive}. It  can be shown that \eqref{N-linear-scheme}  approximates  some linear parabolic heat-like operator when the discrete time step $h$ is small, as we will see in Lemma~\ref{lem:generalhj}.

%Define the normalization constant 
%\begin{equation}\label{Normal-const}
%Z_L(x)=\int_{\mathbb{R}^d} e^{-h\mathcal{L}(x, v)/(2\nu)}dv,
%\end{equation}
%then
%\begin{equation}\label{recovery}
%\tilde{u}^N(x)=Z_L(x)^Nu^N(x).
%\end{equation}

The discrete-time steps of the algorithm in \cite{gomes2007entropy}, following a time-marching process, then proceeds as follows.

\begin{algorithm}[H]
		\caption{The discrete-time (via time-marching) classical algorithm to compute viscosity solution at time-step $N$ \cite{gomes2007entropy}. The output $S_{\nu}^N$ approximates the viscosity solution $S(t=t^N, x)$ to precision $O(\epsilon)$, see Lemma~\ref{lem:gomeserror}.}
		\label{Alg-1}
\begin{enumerate}
    \item Input: $u^0(x)=u(0,x)=e^{-S_0(x)/(2\nu)}$. Choose a small time step $h$ (see Lemma \ref{lem:gomeserror}); 
    \item Solve for $u^n(x)$, for $n=0, \cdots, N$, by the \textit{linear} scheme \eqref{N-linear-scheme};
  %  \item Define $\tilde{u}^N(x)$ by \eqref{recovery};
    \item The solution to $S_\nu^N(x)$ can be recovered by applying the Cole-Hopf transformation: 
   \begin{equation}\label{HJ-recover}
   S_\nu^N(x)=-2\nu \ln u^N(x);
   \end{equation}
    \item The gradient of $S^N_\nu$ is obtained using
    \begin{align}
        \nabla S_\nu^N (x) = -2\nu \nabla {u}^N(x)/{u}^N(x).
    \end{align}
\end{enumerate}
\end{algorithm}

 Intuitively, one can see how the definition of $S_D^n$ from  Eq.~\eqref{eq:sndef}  approximates the solution of the Hamilton-Jacobi equation. We first observe that for small $\nu \ll 1$, a variation of Laplace's method can be applied to the integral, which means that the integral is dominated by the exponential term at its optimal value. From this, it can be proved that \cite{gomes2007entropy}
\begin{align} \label{eq:varGdef}
  S_D^{n+1}(x)= G[S_D^n](x)=\inf_{\gamma} \int (h\mathcal{L}(x,v)+S_D^n(x+hv)+2 \nu \ln \gamma(v)) \gamma(v) dv,
\end{align}
where an entropy regulariser $\gamma(v) \ln \gamma(v)$ is added for $\nu \neq 0$  and $\gamma(v)$ is a probability density over $\mathbb{R}^d$ so $
\int_{\mathbb{R}^d} \gamma(v) dv =1$. The inclusion of an entropy regularisation \cite{JKO} introduces uncertainty in $v$ and the optimal probability distribution $\gamma(v)$ is expected to take the form of a Gibbs state with temperature $\nu$, from well-known arguments related to free energy minimisation. Eq.~\eqref{eq:varGdef} is also a Hopf-Lax-like representation, one that uses an entropy regulariser. Using this definition, one can form a suggestive argument to see how the Hamilton-Jacobi equation emerges. For small $\nu$ and  $h$, using Eq.~\eqref{eq:varGdef}, the time derivative of $S_D$ behaves roughly like $\partial S_D^n/ \partial t \sim (S_D^{n+1}-S_D^n)/h =(G[S_D^n]-S_D^n)/h \sim \inf (\mathcal{L}-p \cdot v)=-\mathcal{H}(x,p)$, where the $p \cdot v$ term comes from the Taylor expansion $S_D^n(x+hv) \sim S_D^n(x)+h\nabla S_D^n \cdot v$ and $p=-\nabla S_D$. This then gives  the Hamilton-Jacobi equation $\partial S_D/ \partial t+\mathcal{H}(x,p)=0$. This is to provide intuition only, and for full details see \cite{gomes2007entropy}.

%For the error analysis of Algorithm~\ref{Alg-1}, see Appendix~\ref{app:background}. From Lemma~\ref{error-3} in Appendix~\ref{app:background}, we have the following result. 

Next we show the error of this algorithm for the case of $K(v)=|v|^2$. Our analysis uses the results from \cite{gomes2007entropy}, in which most of the theory and error analysis was developed for this special $K(v)$. The analysis for more general $K$ is possible but involves highly technical and detailed analysis, which is out of the scope of this article. The error analysis also uses Lemma \ref{lem:viscosityerror}, but we will use the worst case scenario, $\beta=1/2$.

\begin{lemma} \label{lem:gomeserror}
    Let ${S_\nu} \in C^3(\mathbb{R}^d)$, and $K(v)=|v|^2$. Suppose that $h\|\nabla^2 {S_\nu}\|_{L^\infty}(\mathbb{R}^d) $ is smaller than a suitable constant. Then 
    \begin{align} \label{error-10}
   & S^n_\nu(x) -S(t^n, x)=O\left( h^2d^3/\nu+ hd^2+(h/\nu)^{3/2} d^2+ d\nu^{1/2}\right),\\ 
\label{error-a} &   \|\nabla S^n_\nu(x) -\nabla S(t^n, x)\|_{l_2(\mathbb{T}^d)}=O\left( h^2d^{7/2}/\nu^2+ hd^{5/2} /\nu+h^{3/2} d^{5/2}/\nu^{5/2}+ d\nu^{1/2}\right),  
    \end{align}
    To achieve a precision of $\epsilon$, namely,
    \begin{align} \label{precision}
    S^n_\nu(x) -S(t^n, x) =O(\epsilon),
    \end{align}
    one needs to choose
    \begin{align}\label {requirement-10}
    \nu =O\left( (\epsilon/d)^2\right), \quad h=O\left(\epsilon^{8/3}/d^{10/3} \right).
    \end{align}
     Furthermore, to achieve a precision of $\epsilon$ for $\nabla S_\nu^n$, namely,
    \begin{align} \label{precisiongrad}
    \|\nabla S^n_\nu(x) -\nabla S(t^n, x)\|_{l_2(\mathbb{T}^d)} =O(\epsilon),
    \end{align}
     one needs to choose
    \begin{align}\label{requirement-0}
    \nu =O\left( (\epsilon/d)^2\right), \quad h=O\left(\epsilon^4/d^{5} \right).
    \end{align}
\end{lemma}

    \begin{proof}
       \eqref{error-a} is the result of combining Lemma \ref{prop-24} in Appendix \ref{app:background} and  Lemma \ref{lem:viscosityerror}. In \eqref{error-10} and \eqref{precision}, the errors were given in maximum norm.  The error \eqref{error-a} in gradient loses one order of $\nu$ since $\nabla S_\nu^n =O(1/\nu)$, and from maximum norm to $l_2$ norm one increases the error by a factor of $d^{1/2}$. 
       \eqref{requirement-10} is obtained by requiring each term on the right hand side of \eqref{error-10} to be of $O(\epsilon)$. Likewise for \eqref{requirement-0}. 
       \end{proof}

\begin{remark} \cite{gomes2007entropy} did not carry out the analysis for more general $K(v)$. Nevertheless, since the analysis will be similar and the errors come from Taylor expansion of the solutions, up to the same order of derivatives, hence the error dependence on $d$ is expected to be  the same as the special case of $K(v)=|v|^2$.
\end{remark}

The steps in the algorithm above is in discrete-time. However, a continuous-time limit (setting $h \rightarrow 0$) can also be found, which we see later is more suitable for the quantum simulation of PDEs. To derive the continuous-time limit, it is first necessary to ensure that the following linear iterative process is stable 
\begin{align}
    \frac{u^{n+1}(x)-u^n(x)}{h}=\frac{1}{h}(L-\mathbf{1})u^n(x).
\end{align}
This is guaranteed since the operator $L-\mathbf{1}$ has only non-positive spectra (using $\|L[u]\|\leq \|u\|$, Propositions 19 and 20 \cite{gomes2007entropy}). One can then take the continuous-time limit of the discrete scheme to obtain a linear heat-like parabolic differential equation for $u(t,x)=\lim_{h\to 0} u^n(x)$, 
%\begin{align} \label{eq:discreteL}
%    \frac{\partial u^n(t,x)}{\partial t} \approx \frac{L[u^n](x)-u^n(x)}{h}.
%\end{align}
while $S_\nu=\lim_{h\to 0} S_{\nu}$ is the solution to the corresponding viscous Hamilton-Jacobi PDE in Eq.~\eqref{eq:hjviscosity0}, as shown by the next lemma.

\begin{lemma}\cite{gomes2007entropy} \label{lem:generalhj}
   Let $u(t, x)$ solve the following linear parabolic differential equation
   \begin{align} \label{eq:heatcontinuous}
      &  \frac{\partial u(t, x)}{\partial t}=\sum_{i=1}^d \mu_i \frac{\partial u(t,x)}{\partial x_i}+\sum_{j,k=1}^d \nu_{jk} \frac{\partial^2 u(t,x)}{\partial x_j \partial x_k}, \qquad u(0,x)=\frac{e^{-S(0,x)/(2\nu)}}{\mathcal{N}_0}, \qquad \mathcal{N}^2_0=\int e^{-S(0,x)/\nu }dx,\nonumber \\
      & \mu_i=v_i^*, \qquad \nu_{jk}=v^*_j v^*_k+2\nu (D^{-1})_{jk},  \qquad v^*=\text{argmin}_v K(v), \qquad D_{jk}=\partial^2 K(v^*)/\partial v_j \partial v_k,
   \end{align}
   where the drift coefficient $\mu_i$ and the diffusion coeffients $\nu_{ij}$ are all independent of $x$ and $D$ is a matrix with elements $D_{jk}$.  
   Let $S_{\nu}(t,x)$ be the solution of the viscous Hamilton-Jacobi equation in Eq.~\eqref{eq:hjviscosity0} with
   \begin{align} \label{eq:hamgeneral}
     \mathcal{H}(x, \nabla S_{\nu}(t, x))=\sup_v(-p \cdot v-K(v)+V(x)),
   \end{align}
    (assuming $K(v)$ strictly convex in $v$ and superlinear at infinity, $V(x)$ semiconvex and bounded). Then we can approximate $S(t,x)$, the viscosity solution to the Hamilton-Jacobi PDE in Eq.~\eqref{hj0}, by $S_\nu(t,x)$.
   
%   \begin{align}
%       \|S(t,x)+2\nu \ln(u(t,x))\|\leq O(\nu),
%   \end{align}
%   in the case without explicit time-dependence in %$H(\nabla S(t, x), x)$. 
\end{lemma}

%\begin{proof}
%See Appendix~\ref{app:proofgeneralhj}. \\
%\end{proof}
%{\color {green} Here we need to be very careful. At the %formal level, since $\nu \ln{u}=O(\nu)$, you need to get an %error of $O(\nu^2)$ above for the error to make sense. But due to the presence of $\nu$ in Cole-Hopf transformation, some derivatives of $u$ that are ignored in the estimates  may depend on $1/\nu$, so we need very careful analysis.}

%We note there that in the case where the Hamiltonian has no explicit dependence on time, which usually means $U(x)$ has no explicit dependence on time, then the solution approximated by $-2\nu \ln u(t,x)$ has no explicit dependence on the potential $U(x)$ to up $O(\nu)$ precision. This is not surprising when examining  our previous example of classical mechanics, where the potential term $V(t,x)$ in the Hamiltonian-Jacobi term corresponds to a drift term $V(t,x)/(2\nu)$ in the corresponding heat equation. So it means a $O(1)$ drift term implies a $O(\nu)$ potential term in the Hamilton-Jacobi equation, and would thus be negligible if we ignore terms linear in $\nu$ or higher. \\

To obtain the values $\mu_i$ and $\nu_{ij}$ we can assume a form of $K(v)$ where the minimum is simple to extract using classical processing. It is also possible to devise quantum algorithms to obtain $\mu_i$ and $\nu_{ij}$, but we will not examine in detail these settings here.

Next we summarize the 
        continuous-time  version of Algorithm~\ref{Alg-1}, which is what we will use later for the quantum simulation of $u(t,x)$, from which we will estimate observables of $S(t,x)$.

\begin{algorithm}[H]
		\caption{Continuous-time  version of Algorithm~\ref{Alg-1}.}
		\label{Alg-2}
\begin{enumerate}
\item Input: $u(0,x)=u_0(x)=e^{-S_0(x)/(2\nu)}$
\item Solve for $u(t,x)$ using the linear parabolic PDE in Eq.~\eqref{eq:heatcontinuous} (using the choice $\mathcal{N}_0=1$) in Lemma~\ref{lem:generalhj};
%\item {\color {green} Define the continuous normalization constant
%\begin{equation}\label{continous-normalization}
%Z_c(t, x)=\lim_{h\to 0, n\to \infty, nh=t} Z_L(x)^n;
%\end{equation}
\item Apply the Cole-Hopf transformation to $u(t,x)$ to obtain $S_{\nu}(t,x)=-2\nu \ln  u(t,x)$.
%This estimates $S(t,x)$ to precision $O(\nu^{1/2})$, see Remark~\ref{rem:scerror}
\end{enumerate}
\end{algorithm}

\begin{remark}\label{rem:scerror} By Proposition 23 in \cite{gomes2007entropy},
\begin{align}
  u^n - u =O(h^2\|D^3 u^n\|_{L^\infty})=O( d^2h^2), \nonumber 
\end{align}
where we used $\|D^3 u^n\|_{L^\infty}=O(d^2)$. 
 This is the extra error introduced on top of the previous error of the discrete scheme. Since $O(d^2h^2)$ is much smaller than the terms in \eqref{error-1}, the final error for this continuous time scheme Algorithm \ref{Alg-2} is the same order as the discrete scheme stated in Lemma
 \ref{lem:gomeserror}.
\end{remark}

\section{Quantum simulation of the corresponding linear parabolic PDEs}
\label{sec:quantumsimulation}

%{\color {green}  Don't use too many "we"s. Use passive expression like "this gives, this leads to, it yields, etc" or neutral ones like "one gets, one has, etc.".  Especially in theorems and lemmas don't use "we" at all, since it is supposed to be a general statement, not personal ones.

%I deleted or reworded some. But there are way too many "we"s. 

%This is a matter of taste, and my Phd advisor always cares about such things--he got this from Lax. Pay attention to elegancy of writting.  It shows a good taste. 

%"We" should be  used only when  it is really good stuff that shows your ability, or in the time when one really wants to emphasize, like "We assume, We obtain the following estimates", "We finally come to the conclusion, ...".  Not in occasion like "We use triangele inequality". Everyone can use that inequality.  Use "By triangle inequality, one gets..., Or Using triangle inequality gives...".

%Also there are lots of super long sentences. They should be separated by punctuations. I did some but probably there are still many.}

The general linear parabolic PDE in $d$ spatial dimensions ($x=(x_1, \cdots, x_d)\in \mathbb{R}^d$) that encompasses both Eq.~\eqref{eq:heat1} and Eq.~\eqref{eq:heatcontinuous} can be written as 
\begin{align}\label{eq:generalpde}
 &  \frac{\partial u(t,x)}{\partial t}=a(t,x)u+\sum_{j=1}^d b_j \frac{\partial u(t,x)}{\partial x_j}+\sum_{j,k=1}^d c_{jk}\frac{\partial^2 u(t,x)}{\partial x_j \partial x_k}, \qquad u(0,x)=u_0(x), \qquad t\geq 0. 
    \end{align}
When the Hamilton-Jacobi equation is given by Eq.~\eqref{eq:hjviscosity1} (quadratic $\mathcal{H}$), then we want to simulate Eq.~\eqref{eq:generalpde} for  $a(t,x)=V(t,x)/(2\nu)$, $b_j=0$ and $c_{jk}=\nu \delta_{jk}$. When we look at the more general $\mathcal{H}$ in Eq.~\eqref{eq:hamgeneral}, then $a(t,x)=0$, $b_j=\mu_i$, $c_{jk}=\nu_{jk}$ where $\mu_i$ and $\nu_{jk}$ are constants given in Lemma~\ref{lem:generalhj}.\\

The goal of quantum simulation is to prepare quantum states $|u(t)\rangle$ whose amplitudes encode the solution $u(t,x)$. Then in Section~\ref{sec:observables} we describe how to estimate observables of $S(t,x)$ once we can prepare $|u(t)\rangle$. The state $|u(t)\rangle$ can come in two forms: continuous-variable and discrete-variable, depending on whether we choose to keep $x$ continuous or we choose to discretise $x$. If we keep the simulation continuous in space as well as time, we call this analog quantum simulation. If it is discrete in space and time, we call this digital quantum simulation. \\

Below we summarise a simple method -- called Schr\"odingerisation -- which is the only currently known method that can enable quantum simulation of linear PDEs in Eq.~\eqref{eq:generalpde} in both the fully continuous-space, continuous-time language. It can also be used in the discrete-variable language (i.e., simulation on qubit-based systems), where the query complexity of the algorithm can also achieve optimality \cite{optimalschr}. The Schr\"odingerisation protocol also gives a very accessible way of extracting normalisation constants that does not rely on any extra algorithms requiring access to special oracles for state preparation or the linear combination of unitaries method, for instance in \cite{wang2017efficient, chakraborty2018power}. We will make use of this normalisation estimator later on for estimating the minimum value of the viscosity solution of the Hamilton-Jacobi equation later on, for example. The reader can refer to \cite{2023analogPDE, schrprl, schrpra} for more details and justification. \\

However, we note that other methods for simulating $|u(t)\rangle$ using digital quantum simulation can also be used, and the success of the overall methodology for estimating the observables of the viscosity solution of Hamilton-Jacobi equation will still carry through. \\

In the rest of this paper, we will drop the $CV$, $DV$ subscripts whenever it is clear from context whether we are using continuous-variables or discrete-variables. \\

\subsection{Analog quantum simulation}
In the continuous-variable language (i.e. $x$ is continuous and not discretised), each $u(t,x)$ function can be embedded into a normalised  infinite-dimensional vector which we call a continuous-variable quantum state $|u(t)\rangle_{CV}$
\begin{align} \label{eq:ukrep}
    |u(t)\rangle_{CV}=\frac{\mathbf{u}(t)}{\|\mathbf{u}(t)\|}, \qquad \mathbf{u}(t)=\int u(t,x)|x\rangle dx, \quad \|\mathbf{u}(t)\|^2= \int |u(t,x)|^2 dx,  
\end{align}
where $\int=\int_{-\infty}^{\infty}$ unless otherwise specified. It is the aim of the quantum algorithm to prepare the state $|u(t)\rangle$. 

Here the eigenbasis $\{|x\rangle\}_{x \in \mathbb{R}^d}$ over the field of complex numbers spans an infinite-dimensional complex Hilbert space, and describes states consisting of $d$ \textit{qumodes}. Qumodes are quantum analogue of a continuous classical degree of freedom, like position, momentum or energy before being quantised.  A qumode is equipped with observables with a continuous spectrum, such as
the position $\hat{x}$ and momentum $\hat{p}$ observables of a quantum particle, also known as quadratures, where $[\hat{x}, \hat{p}]=i\mathbf{1}_x$. Here their eigenvectors are denoted respectively by $|x\rangle$ and $|p\rangle$ where $\langle x|p\rangle=\exp(ixp)/\sqrt{2\pi}$ and $\int dx |x\rangle \langle x|=\mathbf{1}_x=\int dp |p\rangle \langle p|$. The $|x\rangle$ and $|p\rangle$ eigenstates are known as the position and momentum eigenstates. 
For a system of $d$-qumodes, we can define the position/momentum operator only acting on the $j^{\text{th}}$ mode as 
 \begin{align}
    \hat{p}_j=\mathbf{1}_x^{\otimes j-1}\otimes \hat{p} \otimes \mathbf{1}_x^{\otimes d-j}, \qquad \hat{x}_j=\mathbf{1}_x^{\otimes j-1}\otimes \hat{x} \otimes \mathbf{1}_x^{\otimes d-j}, \qquad [\hat{x}_j,\hat{p}_k]=i\delta_{jk}\mathbf{1}_x.
    \end{align}

 The PDE in Eq.~\eqref{eq:generalpde} can be rewritten as 
    \begin{align} \label{eq:generalA}
        \frac{d\mathbf{u}(t)}{dt}=-i\mathbf{A}(t)\mathbf{u}(t), \quad \mathbf{A}(t)=ia(t, \hat{x})+i\sum_{j=1}^db_j \hat{p}_j-\sum_{j,k=1}^d c_{jk} \hat{p}_j \hat{p}_k. 
    \end{align}
   We summarise the algorithm in Algorithm~\ref{Alg-CVsimulation} (see \cite{2023analogPDE} for more details). One can also use an alternative a Jaynes-Cummings-like Hamiltonian in Algorithm~\ref{Alg-CVsimulation}, where  Hamiltonian can be more accessible in analog quantum systems for certain platforms like circuit QED systems, for example see \cite{hyperbolic2024}. 
    
\begin{algorithm}[H]
		\caption{The 
        continuous-time Schr\"odingerisation algorithm  to prepare $|u(t)\rangle_{CV}$ and also to estimate the normalisation constant $\|\mathbf{u}(t)\|$, see Figure~\ref{fig: Normalisation}.}
		\label{Alg-CVsimulation}
\begin{enumerate}

\item    Input:  (a) $d$ qumode state $|u_0\rangle$; (b)  $1$ qumode ancilla state $|\Xi\rangle=\int e^{-|\xi|}|\xi\rangle d\xi=\int 2/(1+\eta^2)|\eta\rangle d \eta$; (c) quantum Hamiltonian $\mathbf{H}(t)=\mathbf{A}_2(t)\otimes \hat{\eta}+\mathbf{A}_1(t)\otimes \mathbf{1}_{\eta}$, $\mathbf{A}_1(t)=(1/2)(\mathbf{A}(t)+\mathbf{A}^{\dagger}(t))$, $\mathbf{A}_2(t)=(i/2)(\mathbf{A}(t)-\mathbf{A}^{\dagger}(t))$, where $[\hat{\xi}, \hat{\eta}]=i\mathbf{1}_{\eta}$.

    \item Given oracle access to $\mathbf{H}(t)$, where $d\mathbf{v}(t)/dt=-i\mathbf{H}(t)\mathbf{v}(t)$, beginning with initial state $|u_0\rangle|\Xi\rangle$, Prepare $|v(t)\rangle$ with quantum simulation  
    
    \item  Output: $|u(t)\rangle$ when $\mathbf{A}_2<0$

    Measure $|v(t)\rangle$ from step (2) using imperfect measurement operator  $\mathbf{1} \otimes \hat{\Pi}_{\text{imp}}=\int_0^{\infty} g(\xi)|\xi\rangle \langle \xi| d\xi$ (where $g(\xi)$ models the imperfection in the detector and $g(\xi)=1$ for perfect projective measurement).  
    
    Flag success -- output  $|u(t)\rangle$  --  when ancilla mode reads $\xi>0$, otherwise measure again. The probability of success is $p_{succ}=(\int_0^{\infty} e^{-\xi}g(\xi)d\xi\|\mathbf{u}(t)\|/\|\mathbf{u}(0)\|)^2$
    \item Output: $\|\mathbf{u}(t)\|$ when $\mathbf{A}_2<0$

     $p_{succ}$ estimated from step (3). Then $\|\mathbf{u}(t)\|=\sqrt{p_{succ}}\|\mathbf{u}(0)\|/\int_0^{\infty} e^{-\xi}g(\xi)d\xi$
\end{enumerate}
\end{algorithm}

\begin{figure}[h] 
\centering
\includegraphics[width=12cm]{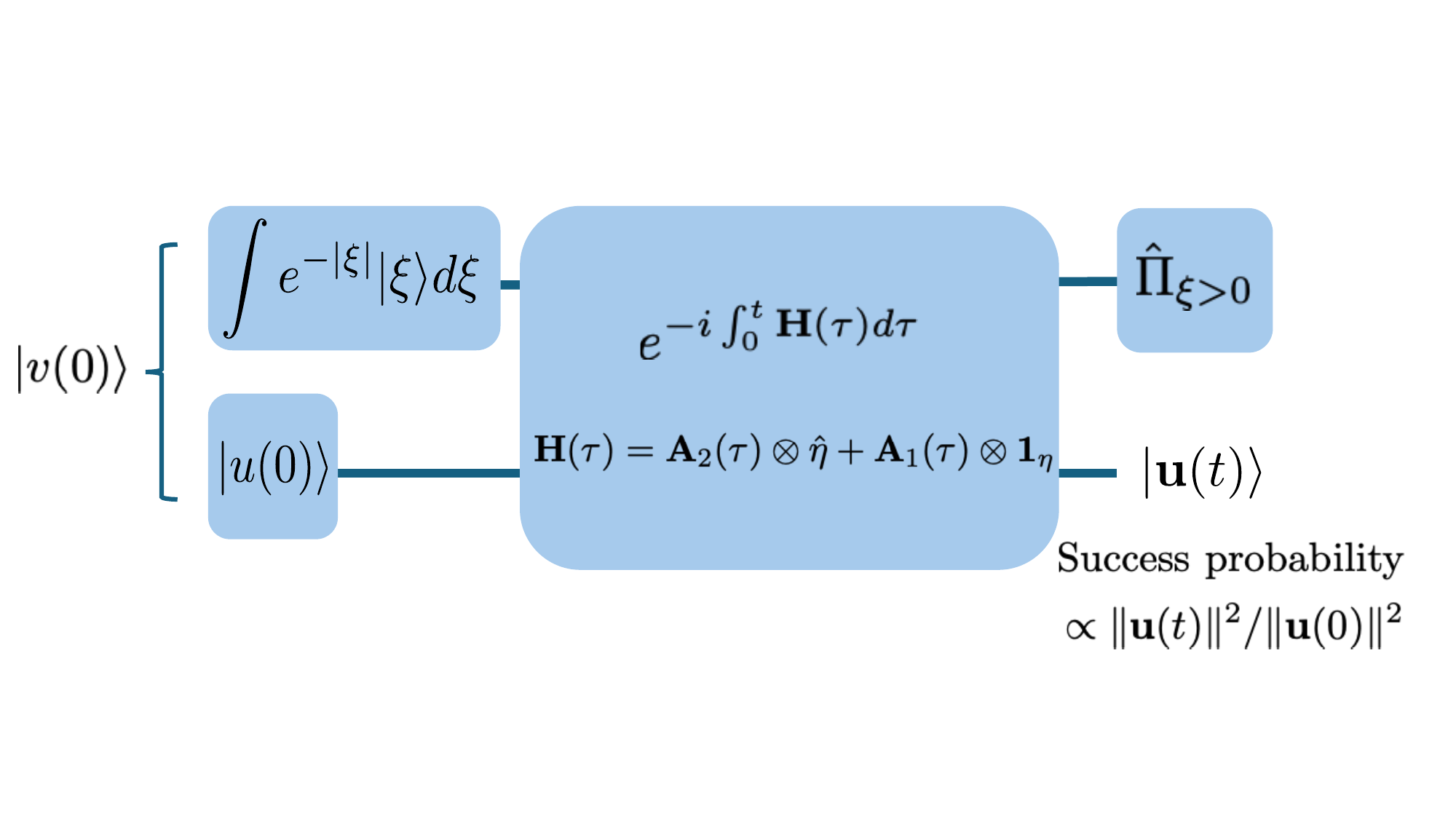} 
%\captionsetup{width=1\textwidth}

\caption{\justifying This is the quantum algorithm for simulating $|u(t)\rangle$ using Schr\"odingerisation, and it is also an algorithm for estimating the normalisation constant $\|\mathbf{u}(t)\|$. This normalisation factor is estimated by measuring the probability when $\xi>0$ is detected for the ancilla state (flag for success). Here the algorithm is expressed in the language of analog quantum simulation for simplicity, and it can easily be extended to the discrete-variable setting by discretising all the states and operators, and in the digital case amplitude amplification can be employed to quadratically boost the success probability.}\label{fig: Normalisation}
\end{figure}

\subsection{Digital quantum simulation}

If we instead embed the solution of our problem into a discrete-variable quantum state $|u(t)\rangle_{DV}$, then we need to first decompose $x \in [-1/2, 1/2]^d$ into a grid of size $N_x$ in each dimension, and consider the cost of simulating the state of $d \ln N_x$ qubits (or a $N_x^d$-dimensional qudit), that embeds the solution of Eq.~\eqref{eq:generalpde}  
\begin{align} \label{eq:digitalu1}
    & |u(t)\rangle_{DV}=\frac{\mathbf{u}(t)}{\|\mathbf{u}(t)\|}, \qquad \mathbf{u}(t,x_j=j/N_x)=\sum_{j=-N_x^d/2}^{N^d_x/2} u_j |j\rangle, \qquad u_j=u(t, j/N_x),\nonumber \\
    & \|\mathbf{u}(t)\|^2=\sum_{j=-N_x^d}^{N^d_x} |u(t, j/N_x)|^2,  \qquad j=j_1\cdots j_d \in \{-N_x/2, N_x/2\}^d.
\end{align}
Each $j_k$, $k=1, \cdots d$ can be decomposed into its binary decimal representation $j_k=j_k^{(1)} \cdots j_k^{(\ln N_x)}$, where each $j_k^{(l)} \in \{0, 1\}$.\\

In Algorithm~\ref{Alg-CVsimulation}, we assumed that all the eigenvalues of $\mathbf{A}_2$ are negative, which means the original PDE is well-posed, which is true for the initial value problems that we consider. The more general case where some eigenvalues of $\mathbf{A}_2$ can be positive can also be tackled. Here Schr\"odingerisation can be used in the exactly the same way along with the same Hamiltonian simulation step for $|v(t)\rangle$, except with a different post-selection process in the final projective measurement. For well-posed problems where all eigenvalues of $\mathbf{A}_2$ is negative, it is sufficient to choose any $\xi^*>0$ and the state becomes a product $\mathbf{v}(t,x)=e^{-\xi^*} \mathbf{u}(t,x)$, so we can retrieve $\mathbf{u}(t)$ easily. However, when $\mathbf{A}_2$ has any positive eigenvalues, we choose $\mathbf{v}(t,x)=e^{-\xi^*} \mathbf{u}(t,x)$ when $\xi^*> \max\{\lambda(\mathbf{A}_2)t, 0\}$ where $\lambda(\mathbf{A}_2)$ represents the eigenvalues of $\mathbf{A}_2$. Then the final projective measurement process if $\Pi_{\xi>\xi^*}$. A prominent example is the backward heat equation, where we have terminal condition instead of an initial condition. We include this case in step (4) in Algorithm~\ref{Alg-DVsimulation} since in the discrete case certain boundary conditions can also give rise to some positive eigenvalues of $\mathbf{A}_2$ even when the PDE itself is well-posed. However, in this paper we do not encounter these problems as we are not concerned about explicit gate constructions for this simulation step and only provide complexity scalings. For more details on the procedure and the justification, see \cite{illposed, Inhomo}.\\

 The analog quantum algorithm in Algorithm~\ref{Alg-CVsimulation} can be very easily extended to a digital quantum algorithm in Algorithm~\ref{Alg-DVsimulation}.  In this case, we  also discretise the $\xi$ and $\eta$ modes into $N_{\xi}$ and $N_{\eta}$ segments respectively. The optimal resource cost is stated in Lemma~\ref{lem:sim}. 
Define
\[ \mathbf{H}_{DV}(t)=\mathbf{A}_{2,DV}(t)\otimes \hat{D}+\mathbf{A}_{1,DV} (t)\otimes I
\]
which is a matrix of size $N_{\eta}N_x^d \times N_{\eta} N_x^d$ that acts on a system of $\ln N_{\eta}+d\ln N_x$ qubits. One can obtain $\mathbf{A}_{1,DV}$ and $\mathbf{A}_{2,DV}$ by discretising the spatial components of $\mathbf{A}_1, \mathbf{A}_2$, for example. In this case replace $\hat{p} \rightarrow \hat{P}$, $\hat{x} \rightarrow{X}$, where $\hat{P}_j$ is a sparse shift matrix and $\hat{X}=\text{diag}(-1/2N_x, \cdots, 1/2N_x)$ for the one-dimensional case, and $\hat{D}=\text{diag}(-1/2N_{\eta}, \cdots, 1/2N_{\eta})$. We summarize the algorithm below. 

 \begin{algorithm}[H]
		\caption{The 
        discrete-variable (digital) quantum algorithm to prepare $|u(t)\rangle_{DV}$  (Schr\"odingerisation) and also to estimate the normalisation constant $\|\mathbf{u}(t)\|_{DV}$.}
		\label{Alg-DVsimulation}
 \begin{enumerate}
 \item Input: discrete-variable counterparts $|u_0\rangle_{DV}$ and $|\Xi\rangle_{DV}$ or smoothed initial ancilla qubit state (see \cite{optimalschr}); oracle access to Hamiltonian  $\mathbf{H}_{DV}(t)$;
 
 \item Hamiltonian simulation for time $t$ applied to state $|v(0)\rangle=|u_0\rangle_{DV} |\Xi\rangle_{DV}$, which evolves with respect to the unitary generated by $\mathbf{H}_{DV}(t)$, to obtain $|v(t)\rangle_{DV}$
 \item Output $|u(t)\rangle$ and $\|\mathbf{u}(t)\|$ when $\mathbf{A}_2<0$. 

Apply projective operator $\hat{\Pi}_{\xi>0}$ on $\ln N_{\xi}$ ancilla qubits (or $N_{\xi}$-dimensional ancilla qudit). Flag success -- output $|u(t)\rangle$ -- when $\xi>0$ and success probability $p_{succ}=\|\mathbf{u}(t)\|/\|\mathbf{u}(0)\|$ (amplitude amplification used for quadratic boosting). 

Output normalisation  $\|\mathbf{u}(t)\|=p_{succ}\|\mathbf{u}(0)\|$.
 
\item Output $|u(t)\rangle$ and $\|\mathbf{u}(t)\|$ when some eigenvalues of $\mathbf{A}_2$ are positive  
    
    Repeat step (3) using the measurement operator $\hat{\Pi}_{\xi>\xi^*}$, where success is flagged when $\xi>\xi^*$ and $\xi^*> \max\{\lambda(\mathbf{A}_2)t, 0\}$. Here $\lambda(\mathbf{A}_2)$ represents the eigenvalues of $\mathbf{A}_2$. 
\end{enumerate}
\end{algorithm}

Suppose we solve the Schr\"odingerized linear heat equation with forward Euler in time, center difference in space, and spectral method in $\eta$, and choose the optimal initial data as in  \cite{optimalschr}. This will introduce an error of $O(h+ d (\Delta x)^2)$ (the error from the discretization in $\eta$ is ignored since it is much smaller when one uses spectral method with smooth data).   Numerical stability condition requires $h=O((\Delta x)^2/d)$ (the above order could  be divided by $\nu$ when $L(v)=|v|^2$, since $c_{jk}=O(\nu)$, but here when we treat the general case we don't seek the sharp estimate, which can be improved when one works on specific examples). So the overall error is of $O(d(\Delta x)^2)$. Combining this with the continuous error results in Lemma \ref{lem:gomeserror}, we have the following error estimates by classical numerical analysis:

\begin{lemma}\label{error-dis}
Let $S_{\nu,\Delta}=-2\nu \ln |u\rangle$   where $|u\rangle$ is the discrete solution obtained by Algorithm \ref{Alg-DVsimulation}. Then,

  \begin{align} \label{error-ad}
   & S^n_{\nu,\Delta}(x) -S(t^n, x)=O\left( h^2d^3/\nu+ hd^2+(h/\nu)^{3/2} d^2+d(\Delta x)^2+ d\nu^{1/2} \right),\\
   \label{error-cd}
  & S^n_{\nu,\Delta}(x) -S_\nu(t^n, x)=O\left( h^2d^3/\nu+ hd^2+(h/\nu)^{3/2} d^2+ d(\Delta x)^2 \right),\\
  \label{error-bd}
 & \|\nabla S^n_{\nu,\Delta}(x) -\nabla S_\nu(t^n, x)\|_{l_2(\mathbb{T}^n)}=O\left( h^{3/2}d^{3}/\nu+ h^{1/2}d^2 +h d^{2}/\nu^{3/2}+ d^{3/2}\Delta x \right). 
    \end{align}
    To achieve a precision of $\epsilon$, namely,
    \begin{align} \label{precisiondelta}
    S^n_{\nu,\Delta} (x) -S(t^n, x) =O(\epsilon), \quad S^n_{\nu,\Delta}(x) -S_\nu(t^n, x)=O(\epsilon),
    \end{align}
    one needs to choose
    \begin{align}\label {requirement-aa}
    \nu =O\left( (\epsilon/d)^2\right), \quad \Delta x=O((\epsilon/d)^{1/2}), \quad h=O\left(\epsilon^{8/3}/d^{10/3} \right).
    \end{align}
     Furthermore, to achieve a precision of $\epsilon$ for $\nabla S_{\nu,\Delta}^n$, namely,
    \begin{align} \label{precisiondeltagrad}
   \| \nabla S^n_{\nu, \Delta}(x) -\nabla S_{\nu}(t^n, x)\|_{l_2(\mathbb{T}^n)} =O(\epsilon),
    \end{align}
     one needs to choose
    \begin{align}\label {requirement-bb}
    \Delta x=O(\epsilon/d^{3/2}), \quad h=O\left(\epsilon^4/d^{5} \right).
    \end{align}

\end{lemma}

\begin{proof} 
The approximation of spatial derivative is one order less accurate in $\Delta x=O((dh)^{1/2})$ than the approximation of the function. In addition, from maximum norm to $l_2$ norm increases the error by a factor of $d^{1/2}$. This is how we get \eqref{error-bd} from \eqref{error-cd}. \eqref{requirement-aa} is obtained by requiring each term on the right hand side of \eqref{error-ad} to be of $O(\epsilon)$. Likewise for \eqref{requirement-bb}.
\end{proof}

\begin{lemma} \label{lem:sim} \cite{schrprl, schrpra, optimalschr} Let the discretisation mesh size in $\eta$ be $\Delta \eta=1/N_{\eta}=O(\mu)$. Then the Schr\"odingerisation algorithm can prepare the state $|u(t)\rangle$ to precision $\epsilon$ with $\Omega(1)$ success probability and a flag indicating success, with query complexity (oracle access to block-encoding of the Hamiltonian $\mathbf{H}_{DV}(t)$, see \cite{optimalschr})
\begin{align}
    \tilde{O}\left(\frac{1}{\|\mathbf{u}(t)\|}\left(\alpha_A t \mu+\ln \left(\frac{\mu}{\epsilon \|\mathbf{u}(t)\|}\right)\right)\right), \qquad \|\mathbf{u}(0)\|=1,
\end{align}
where $\alpha_A \geq \|\mathbf{A}_1\|, \|\mathbf{A}_2\|$ (spectral norm) and using $O(\|\mathbf{u}(0)\|/\|\mathbf{u}(t)\|)$ queries to the initial state preparation protocol for $|v(0)\rangle$. If one discretises the initial ancilla state $|\Xi\rangle_{DV}$ by discretising $\xi$ with the total number of grids $N_{\xi} \sim 1/\epsilon$, then $\mu=O(1/\epsilon)$. There also exists optimal choices of the ancilla initial state that can give $\mu=O(\ln(1/\epsilon))$. For time-dependent Hamiltonians, the optimal query complexity only incurs an extra $\ln(1/\epsilon)$ factor. 
\end{lemma}

\begin{corollary}
Ignoring the terms with logarithmic scaling, one can simulate $|u(t)\rangle_{DV}$ to precision $\epsilon$ with cost $\tilde{O}(\alpha_A t \ln(1/\epsilon)/\|\mathbf{u}(t)\|)$, where it is sufficient to choose $\alpha_A \geq \max \{\|\mathbf{A}_1\|, \|\mathbf{A}_2\| \}=O(d^4/\epsilon^2)$.  
\end{corollary}
\begin{proof}
If we use centre difference approximation to approximate the space derivative in the linear parabolic equation \eqref{eq:generalpde}, we can obtain the corresponding $\mathbf{A}_1$ and $\mathbf{A}_2$, 
then $\mathbf{A}_1$ is the difference operator corresponding to the first order derivative term while $\mathbf{A}_2$ corresponds to the other two terms. It is classical results that
\begin{equation}\label{norm-est}
\|\mathbf{A}_1\| = O(dN_x), \quad
\|\mathbf{A}_2\| = O(dN_x^2).
\end{equation}
Since we take $\Delta x=O(1/N_x)=O((\epsilon/d)^{1/2})$ in \eqref{requirement-aa} to approximate $S$, this gives
\begin{align}\label{norm-est}
\|\mathbf{A}_1\| = O(d^{3/2}/\epsilon^{1/2}), \quad
\|\mathbf{A}_2\| = O(d^2/\epsilon).
\end{align}
If one wants to compute $\nabla S$, then \eqref{requirement-bb} requires $\Delta x=1/N_x=O(\epsilon/d^{3/2})$.
 Consequently,
 \begin{align}\label{norm-dest}
\|\mathbf{A}_1\| = O(d^{5/2}/\epsilon), \quad
\|\mathbf{A}_2\| = O(d^4/\epsilon^2).
\end{align}
\end{proof}
\begin{remark}
   For example, if $L(v)=|v|^2$, then $b_j=0, c_{jk}=O(\nu)$ in \eqref{eq:generalpde}. Consequently $\|\mathbf{A}_1\|=O(1)$, while $\|\mathbf{A}_2\|$   will be multiplied by $\nu=O((\epsilon/d)^2)$.
\end{remark}

\begin{remark}
    Here we have given an example of a quantum simulation algorithm when we discretise the continuous-time algorithm in Algorithm~\ref{Alg-2}, by discretising Eq.~\eqref{eq:heatcontinuous}. However, it is also possible to directly use Algorithm~\ref{Alg-1} via the linear scheme in~\eqref{N-linear-scheme}, and we apply instead a time-marching quantum simulation algorithm to approximate $|u(t)\rangle_{DV}$. The comparison between these methods and more precise implementation under different boundary conditions will be explored in future work. 
\end{remark}
The above is the cost for the optimal algorithm and it involves block-encoding oracles and a different initial ancilla state from the simplest one stated. However, there are simpler Hamiltonian simulation algorithms (for example using sparse-access instead of block-encoding) one can also use and the simplest initial ancilla state gives the first-order Schr\"odingerisation scheme. The cost in these cases is also not prohibitive, see Appendix~\ref{app:simplercostly}.

\section{Estimating physically relevant quantities of the Hamilton-Jacobi equation} \label{sec:observables}

In the previous section, it is easy to see that we can always use quantum simulation to simulate  $|u(t)\rangle$. However, what we are really interested in are physically relevant quantities of the viscosity solution of the original Hamilton-Jacobi equation in Eq.~\eqref{hj0}. Since full tomography of the output state is extremely inefficient in general, we must find alternative methods. It can be non-trivial to find suitable measurement protocols to extract those quantities of interest.  \\

In this section, we construct quantum protocols to extract four quantities of interest 

\begin{enumerate}
      \item the value of $S$ at a point,  
   
    \item the gradient $\nabla S$ at a point, 
    
    \item  the minimum value of $S$, 
    
     \item the value of a known function $f(x)$ at the location of the minimum of $S$.  
    \end{enumerate}

The gradient $\nabla S(t,x)$ corresponds to physically-relevant observables like  velocity in the Hamilton-Jacobi PDE, and $|\nabla S|^2/2$ would correspond to a kinetic energy. The gradient is in fact also the solution to the forced Burgers' equation \eqref{burgers} in Eq.~\ref{burgers} since here the gradient $\nabla S$ is by definition curl-free (unlike the usual Burger's equation which is only curl-free in 1D). \\

In many applications of the Hamilton-Jacobi equation, for example in geometric optics and semi-classical computation of quantum dynamics (the WKB analysis), the minimum value of $S(t, x)$ for a given $t$ is of interest since it corresponds to stationary phase, which is an important quantity when one wants to evaluate a highly oscillatory integral \cite{evans2022partial}.  While it would be inefficient to extract all the values of $S(t, x_a)$ in order to identify the minimum value, we can identify a quantum algorithm that can directly estimate this quantity.\\

Here a function $f(x)$ can represent any cost function one is interested in at the stationary point $x^*$ when $S$ is at its minimum and its gradient vanishes. For example, to obtain a total energy, one could be interested in a potential energy contribution $f(x)=V(x)$ at $x^*$.\\

For each of the four quantities we mentioned above, we construct both analog and digital quantum protocols. See Table I for the list of protocols. \\

\begin{table}[h!]
\centering
    \large % or \large
\caption{List of algorithms for simulating $|u(t)\rangle$ and for estimating different quantities . Here we denote $x^*=\text{argmin}_x S(t,x)$ at some fixed $t>0$. All these algorithms require calls to either Algorithm~\ref{Alg-CVsimulation} (or Algorithm~\ref{Alg-DVsimulation}) for the preparation of $|u(t)\rangle_{CV}$ (or $|u(t)\rangle_{DV}$) or the estimation of the normalisation constant $\|\mathbf{u}(t)\|^2$.}
\begin{tabular}{|l|l|l|ll}
\cline{1-3}
State or quantity  & Analog protocol & Digital protocol   \\ \cline{1-3}
%Classical algorithm for $\mathbf{u}(t)$ & Algorithm~\ref{Alg-1} (continuous time) & Algorithm~\ref{Alg-2} (discrete time)\\ \cline{1-3}
$|u(t)\rangle$, \quad $\|\mathbf{u}(t)\|^2$ & Algorithm~\ref{Alg-CVsimulation} & Algorithm~\ref{Alg-DVsimulation} \\ \cline{1-3}
$S(t, x_a)$ at some $x=x_a$ & Algorithm~\ref{Alg-CVSx} & Algorithm~\ref{Alg-DVSx} \\ \cline{1-3}
$\partial_k S(x_a)$ at some $x=x_a$  & Algorithm~\ref{Alg-CVgradS} & Algorithm~\ref{Alg-DVgradS} \\ \cline{1-3}
$S_{\min}(t)=\min_x S(t,x)$   & Algorithm~\ref{Alg-CVSmin} & Algorithm~\ref{Alg-DVSmin}  \\ \cline{1-3}
$f(x^*)$ for some given function $f$ & Algorithm~\ref{Alg-CVfmin} & Algorithm~\ref{Alg-DVfmin} \\ \cline{1-3}
\end{tabular}
\end{table} \label{tab:list}

\subsection{Estimating the value at a point}
We first observe that at a  location $x=x_a$, if we measure the probability $|\langle x_a|u(t)\rangle|^2$, then we can write 
\begin{align}
   - \nu \ln|\langle x_a|u(t)\rangle|^2=-\nu \ln \left(\frac{e^{-S_{\nu}(t, x_a)/\nu}}{\|\mathbf{u}(t)\|^2}\right)= S_{\nu}(t, x_a)+\nu \ln \|\mathbf{u}(t)\|^2. 
\end{align}
This means the viscosity solution $S(t, x_a)$ can be estimated from measuring the normalisation constant and the probability $|\langle x_a|u(t)\rangle|^2$:
\begin{align}
    |S(t, x_a)-S_{\nu}(t, x_a)|=|S(t, x_a)+(\nu \ln|\langle x_a|u(t)\rangle|^2+\nu \ln \|\mathbf{u}(t)\|^2)| \leq \epsilon_{\min}, \qquad \nu=O(\epsilon^2_{\min}/d^2),
\end{align}
 where $\epsilon_{\min}$ is the minimum error if the normalisation constant is known and $|u(t)\rangle$ is prepared exactly and we use the result in Lemma~\ref{lem:gomeserror}. Thus using Algorithm~\ref{Alg-CVsimulation} or Algorithm~\ref{Alg-DVsimulation} to prepare $|u(t)\rangle$ and to measure its normalisation constant, it is then straightforward to estimate $S(t, x_a)$ at a particular location $x=x_a$. See Algorithm~\ref{Alg-CVSx} for a summary and the cost in Lemma~\ref{lem:cvsx}. 

\subsubsection{Analog protocol to measure value at a point} 
\begin{algorithm}[H]
		\caption{An analog protocol to estimate $S(t,x)$ at some given point $x=x_a$. Output estimates $S(t, x_a)$ to precision $\epsilon$ with cost given in Lemma~\ref{lem:cvsx}.}
		\label{Alg-CVSx}
\begin{enumerate}
    \item Input:  $x_a$ and $|u(t)\rangle_{CV}$, $\|\mathbf{u}(t)\|_{CV}^2$ from Algorithm~\ref{Alg-CVsimulation} 
    \item Output: $- \nu \ln|\langle x_a|u(t)\rangle_{CV}|^2-\nu \ln \|\mathbf{u}(t)\|_{CV}^2$ 
\end{enumerate}
\end{algorithm}
 Although here we use the projection of $|u(t)\rangle$ onto outcome $|x\rangle=|x_a\rangle$ after applying a projective measurement $|x\rangle \langle x|$, which is an ideal measurement, this can be straightforwardly generalised to more realistic measurements as well. 
 
\begin{lemma} \label{lem:cvsx}
    The minimum precision in estimating $S(t, x_a)$ with using $|u(t)\rangle$ and the normalisation constant of $S_{\nu}(t, x_a)$ is $\epsilon_{\min}$. Then, using Algorithm~\ref{Alg-CVSx}, the number of preparations of state $|v(t)\rangle$ in Step (2) of Algorithm~\ref{Alg-CVsimulation} required to estimate $S(t, x_a)$ to precision $\epsilon$ with success probability $1-\delta$ is $O(\ln(1/\delta) \nu^2 /(\epsilon-\epsilon_{\min})^2)$, and the constant $\nu$ in Algorithm~\ref{Alg-CVsimulation} can be chosen to be $O(\epsilon^2_{\min}/d^2)$.  
\end{lemma}
\begin{proof}
Here we have two sources of error: one is the error $\epsilon_{j}$ in measuring the probability $|\langle x_a|u(t)\rangle|^2$ and the other error $\epsilon_{norm}$ in measuring $\|\mathbf{u}(t)\|^2$, which we have already seen is also a probability of obtaining $\xi>0$ in the ancilla qumode. This means the corresponding errors in $-\nu \ln |\langle x_a|u(t)\rangle|^2$ and $-\nu \ln \|\mathbf{u}(t)\|^2$ are respectively $\tilde{\epsilon}_{j}\sim \nu \epsilon_j/|\langle x_a|u(t)\rangle|^2$ and $\tilde{\epsilon}_{norm} \sim \nu \epsilon_{norm}/ \|\mathbf{u}(t)\|^2$ respectively. So the total error in $S(t,x_a)$ becomes 
\begin{align} \label{eq:cvsineq}
     |S(t, x_a)-\tilde{S}_{\nu}(t, x_a)|\leq |S(t, x_a)-S_{\nu}(t, x_a)|+|S_{\nu}(t, x_a)-\tilde{S}_{\nu}(t, x_a)|\leq \epsilon_{\min}+\tilde{\epsilon}_j+\tilde{\epsilon}_{norm} \sim \epsilon.
\end{align}
Let the number of $|u(t)\rangle$ samples used  be denoted $\mathcal{N}_j\sim \ln (1/\delta)/\epsilon_j^2$ and the number of quantum simulation protocol with respect to $\mathbf{H}(t)$ be denoted $\mathcal{N}_{norm}\sim \ln(1/\delta)/\epsilon^2_{norm}$. Since each preparation of $|u(t)\rangle$ already provides information about $\|\mathbf{u}(t)\|^2$, then clearly $\mathcal{N}_{norm} \sim \mathcal{N}_j/\|\mathbf{u}(t)\|^2$, since $O(\|\mathbf{u}(t)\|^2)$ is the success probability of getting $|u(t)\rangle$ from the quantum simulation algorithm with respect to $\mathbf{H}(t)$, so $\epsilon^2_{norm}\sim \epsilon^2_j \|\mathbf{u}(t)\|^2$. Inserting into Eq.~\eqref{eq:cvsineq} implies $\nu \epsilon_j(1/|\langle x_a|u(t)\rangle|^2+1/\|\mathbf{u}(t)\|) \sim \epsilon-\epsilon_{\min}$. A sufficient total number of quantum simulation algorithms with respect to $\mathbf{H}(t)$ required is therefore $\mathcal{N}_{norm}\sim \ln (1/\delta)/(\epsilon_j^2 \|\mathbf{u}(t)\|^2)$, where $\epsilon_j$ can be constrained from the above relationship. So the total cost is $\sim \ln(1/\delta) \nu^2 J/(\epsilon-\epsilon_{\min})^2$, where $J=(1/|\langle x_a|u(t)\rangle|^2+1/\|\mathbf{u}(t)\|)^2/\|\mathbf{u}(t)\|^2$. Then from Lemma~\ref{lem:gomeserror}, we see that we can choose $\nu=O(\epsilon^2_{\min}/d^2)$.

\end{proof}

\subsubsection{Digital protocol to measure value at a point} \label{sec:DVobservables}
\begin{algorithm}[H]
		\caption{Summary of digital protocol to estimate $S(t,x)$ at some given point $x=x_a$. Output estimates $S(t, x_a)$ to precision $\epsilon$ with cost given in Lemma~\ref{lem:dvsx}.}
		\label{Alg-DVSx}
\begin{enumerate}
    \item Input:  $x_a$ and $|u(t)\rangle_{DV}$, , $\|\mathbf{u}(t)\|_{DV}^2$ from Algorithm~\ref{Alg-DVsimulation} 
    \item Output: $- \nu \ln|\langle x_a|u(t)\rangle_{DV}|^2-\nu \ln \|\mathbf{u}(t)\|_{DV}^2$.
\end{enumerate}
\end{algorithm}
Here the digital protocol follows the same idea as the analog protocol, except we need to use the digital state $|u(t)\rangle_{DV}$ and its corresponding normalisation constant. Here $|u(t)\rangle_{DV}$ results from a discretised PDE, so we need to modify $\epsilon_{\min}$ to include the error from discretisation. 

\begin{lemma} \label{lem:dvsx}
 The minimum precision in estimating $S(t, x_a)$ when using $|u(t)\rangle$ constructed from $d\ln N_x$ qubits and the normalisation constant of $S_{\nu}(t, x_a)$ is $\epsilon_{\min}$, where a large enough $N_x$ is chosen so $1/N_x=O\left(\sqrt{\epsilon_{\min}/d}\right)$. Then, using Algorithm~\ref{Alg-DVSx}, the number of preparations of state $|v(t)\rangle$ in Step (2) of Algorithm~\ref{Alg-DVsimulation} required to estimate $S(t, x_a)$ to precision $\epsilon$ with success probability $1-\delta$ is $O(\ln(1/\delta) \nu^2 /(\epsilon-\epsilon_{\min})^2)$, and it is sufficient to choose the constant $\nu$ Algorithm~\ref{Alg-DVsimulation} to be $O(\epsilon^2_{\min}/d^2)$. 
\end{lemma}
\begin{proof}
   The proof follows in the same way as Lemma~\ref{lem:cvsx}, where the main difference is that we also need to take into account error coming from discretisation using finite $N_x$. Then $S_{\nu, \Delta}(t, x_a)=-\nu \ln|\langle x_a|u(t)\rangle_{DV}|^2-\nu \ln \|\mathbf{u}(t)\|_{DV}^2$ is the estimate for $S_{\nu}$ when using the solution $u$ from the discretised parabolic PDE. Using Lemma~\ref{error-dis} 
   \begin{align}
   |S(t, x_a)-S_{\nu, \Delta}(t, x_a)|
     \leq \epsilon_{\min}, \qquad \nu=O(\epsilon^2_{\min}/d^2), \qquad (1/N_x)=O\left(\sqrt{\epsilon_{\min}/d}\right). 
\end{align}
This means the total error 
\begin{align}
    |S(t, x_a)-\tilde{S}_{\nu}(t, x_a)|\leq  |S(t, x_a)-S_{\nu, \Delta}(t, x_a)|+ |S_{\nu, \Delta}(t, x_a)-\tilde{S}_{\nu}(t, x_a)|\leq \epsilon_{\min}+\zeta=\epsilon
\end{align}
where the error $\zeta$ comes from sampling. The cost to estimate $S(t,x_a)$ has the same form as the analog protocol $O(\ln(1/\delta) \nu^2 /(\epsilon-\epsilon_{\min})^2)$. 
\end{proof}

\subsection{Estimating the gradient at a point} 

If given copies of $|u(t)\rangle$, there are different ways to estimate the gradient $ \partial S_{\nu}(t, x_a)/\partial x_k$ at a particular point $x=x_a$. Here we choose a more elegant method based on weak measurement, which involves measuring an ancilla or probe state and using a 'pointer-variable' measurement to output the final observable. We examine both the analog and digital protocols below.

\subsubsection{Analog protocol for gradient estimation} \label{sec:CVgradient}

We summarise the analog protocol in Algorithm~\ref{Alg-CVgradS} with cost given in Lemma~\ref{lem:CVgradient}.

\begin{algorithm}[H]
		\caption{Summary of analog protocol to estimate $|g_{ka}| \equiv \partial _kS_{\nu}(t,x_a)$, see Fig.~\ref{fig:CVgradient}. Output estimates $\partial_k S(t, x_a)$ to precision $\epsilon$ with cost given in  Lemma~\ref{lem:CVgradient}.}
\label{Alg-CVgradS}
\begin{enumerate}
    \item Input:   $|u(t)\rangle_{CV}$ from Algorithm~\ref{Alg-CVsimulation};  copies of single-mode vacuum coherent state $|\alpha_0\rangle$ (ancilla qumode); ability to perform unitary operation $\exp(2i  \kappa \nu \hat{p}_k \otimes \hat{p})$, $x_a$;
    \item Apply $\exp(2i \kappa \nu \hat{p}_k \otimes \hat{p})$ to $|u(t)\rangle_{CV}|\alpha_0\rangle$, where $\kappa \ll 1$ makes this more experimentally feasible;
    \item Project the $d$-mode (non-ancilla) state onto $|x_a\rangle \langle x_a|$. This has success probability $\phi_{coh}(x_a)$;
    \item  If step (3) succeeds, the resulting ancilla mode is $|\Psi\rangle$. Measure  $\langle \Psi|\hat{x}|\Psi\rangle$ and $\langle \Psi|\hat{p}|\Psi \rangle$;
    \item Output: $|g_{ka}|=\sqrt{\langle \Psi|\hat{x}|\Psi\rangle^2+\langle \Psi|\hat{p}|\Psi\rangle^2}/\kappa$
\end{enumerate}
\end{algorithm}

\begin{figure}[h] 
\centering
\includegraphics[width=12cm]{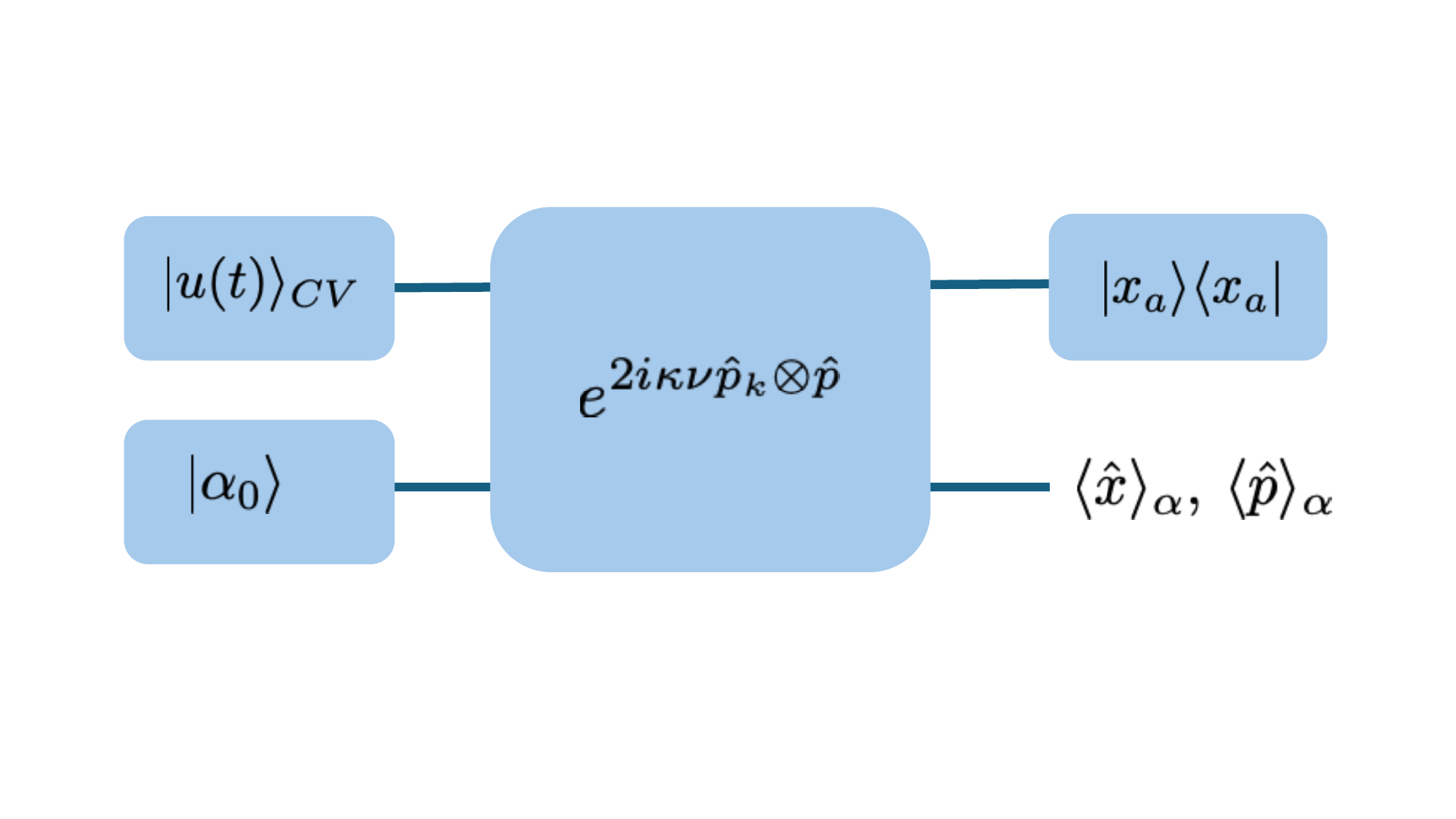} 
%\captionsetup{width=1\textwidth}

\caption{\justifying Schematic diagram of analog quantum protocol for extracting the gradient $|g_{ka}|=\partial S_{\nu}(t, x_a)/\partial x_k$, when given copies of the coherent state $|\alpha_0\rangle$ (easy to prepare) and the state $|u(t)\rangle$. Here $u(t,x)$ is related to $S_{\nu}(t,x)$ by the Cole-Hopf transformation as outlined in ~\ref{sec:heatquantumsimulation}. The coherent state is an ancilla qumode that acts as a measurement probe, and we proceed to use a 'pointer-variable' measurement. Given the initial state $|u(t)\rangle|\alpha_0\rangle$, we apply a Gaussian operation $\exp(2i\kappa\nu \hat{p}_k \otimes \hat{p})$, where $\hat{p}_k$ acts on the $|u(t)\rangle$ register and $\hat{p}$ acts on the coherent state register and $\kappa \ll 1$. Then projecting the resulting first register onto the desired position $|x_a\rangle \langle x_a|$, the resulting ancilla qumode will have its coherence shifted by a value that is $-i$ times $\partial S_{\nu}(t, x_a)/\partial x_k$. To extract this coherence, the expectation value of the ancilla mode with respect to $\hat{x}$ and $\hat{p}$ are taken. Then we can estimate the real and imaginary components of $g_{ak}$ by measuring the quadratures $\langle \hat{x}\rangle_{\alpha}=\kappa\text{Im}(g_{ka})$, $\langle \hat{p}\rangle_{\alpha}=-\kappa \text{Re}(g_{ka})$, and thus recover $|g_{ka}|$.}\label{fig:CVgradient}
\end{figure}
   We first begin with the following simple observation 
    \begin{align} \label{eq:gradS1}
        |g_{ka}| \equiv \frac{\partial S_{\nu}(t,x_a)}{\partial x_k}=-2 \nu \frac{1}{u(t,x_a)}\frac{\partial u(t,x_a)}{\partial x_k}, \qquad g_{ka}=-2i \nu \frac{\langle x_a|\hat{p}_k|u(t)\rangle}{\langle x_a|u(t)\rangle},
    \end{align}
 where $|g_{ka}|$ is real-valued. If the amplitudes of $|u(t)\rangle$ are always real-valued, then $|g_{ka}|=g_{ka}$, however, complex phases could appear in the amplitudes during quantum simulation. To retrieve the real solutions of the original Hamilton-Jacobi equation, the absolute value is taken. 
Then $g_{ka}$ can be estimated with the help of a vacuum coherent state, which we label $|\alpha_0\rangle$. Apply the unitary operation $\exp(2i\delta \nu \hat{p}_k \otimes \hat{p})$ onto $|u(t)\rangle |\alpha_0\rangle$ where the second $\hat{p}$ acts on the ancilla mode, for $\kappa \ll 1$ we can write
\begin{align} \label{eq:CVgenerica}
   &  |x_a\rangle \langle x_a|e^{2 i\kappa\nu \hat{p}_k \otimes \hat{p}}|u(t)\rangle |\alpha_0\rangle=|x_a\rangle\langle x_a|u(t)\rangle\left(\mathbf{1}-\kappa g_{ka}\hat{p}\right)|\alpha_0\rangle+O(\kappa^2)=|x_a\rangle\langle x_a|u(t)\rangle e^{-i\kappa g_{ka}\hat{p}}|\alpha_0\rangle+O(\kappa^2). 
\end{align}
We examine only the ancilla state after the projective measurement, which now has the form $|\Psi\rangle$, so up to $O(\delta^2)$
\begin{align}
 &  |x_a\rangle \langle x_a|e^{2i \kappa \nu \hat{p}_k \otimes \hat{p}}|u(t)\rangle |\alpha_0\rangle \propto |x_a\rangle |\Psi\rangle, \nonumber \\
 &|\Psi\rangle= \frac{e^{-\kappa g_{ka}\hat{p}}|\alpha_0\rangle}{\|e^{- \kappa g_{ka}\hat{p}}|\alpha_0\rangle\|}, \qquad \|e^{-\kappa g_{ka}\hat{p}}|\alpha=0\rangle\|^2=e^{\kappa^2\text{Re}( g_{ka})^2},
\end{align}
where the normalisation can be easily verified using $|\alpha_0\rangle=(1/\pi)^{1/4}\int e^{-p^2/2}|p\rangle dp$. Then it can be straightforwardly verified that 
\begin{align}
& \langle \Psi|\hat{x}|\Psi\rangle = \kappa \text{Im} (g_{ka}), \qquad 
    \langle \Psi|\hat{p}|\Psi\rangle =-\kappa \text{Re}(g_{ka}).
\end{align}
Thus to retrieve the gradient we use
\begin{align} \label{eq:cvgka}
    |g_{ka}|=\sqrt{\langle \Psi|\hat{x}|\Psi\rangle^2+\langle \Psi|\hat{p}|\Psi\rangle^2}/\kappa.
\end{align}

To prepare $|\Psi\rangle$, we require the application of a Gaussian two-mode entangling gate (the continuous-variable controlled-phase (CZ) gate) $\exp(2i \kappa\nu \hat{p}_k \otimes \hat{p})$ with small $\nu \ll 1$ coupling constant, followed by a projective measurement onto $|x_a\rangle$. This projective measurement has success probability
\begin{align}
   \text{Tr}\left((\mathbf{1} \otimes |x_a\rangle \langle x_a|)e^{2i \kappa \nu \hat{p}_k \otimes \hat{p}}(|u(t)\rangle \langle u(t)| \otimes |\alpha_0\rangle \langle \alpha_0|)e^{-2i \kappa \nu \hat{p}_k \otimes \hat{p}}\right)=\int \frac{e^{-p^2}}{\sqrt{\pi}}\frac{|u(x_a+p)|^2}{\|\mathbf{u}(t)\|^2} dp=\phi_{coh}(x^{j}).
\end{align}
This success probability can also be easily boosted to a value as close to $|u(t, x_a)|^2/\|\mathbf{u}(t)\|^2$ as needed. For example, instead of using the coherent state input, one can instead choose a state with much less variance in one quadrature direction compared to another and also changing $\exp(2i \kappa \nu \hat{p}_k \otimes \hat{p})$ to $\exp(2i \kappa \nu \hat{p}_k \otimes \hat{x})$. For example, a state highly squeezed in the $\hat{x}$ quadrature would provide much less variance in recovering $\text{Im}(g_{jk})$. However, such an initial state preparation would be more difficult than preparing a coherent state, and the payoff may not be so high since improvement in $\text{Im}(g_{jk})$ estimation can always be improved by increasing sampling size instead of changing the ancilla state.

\begin{lemma} \label{lem:CVgradient}
    The minimum possible precision to recover $\partial S(t, x_a)/\partial x_k$ using $\partial S_{\nu}(t, x_a)/\partial x_k$ is $\epsilon_{\min}$. Then using Algorithm~\ref{Alg-CVgradS} with $\nu \sim O(\epsilon^2_{\min}/d^2)$ and with success probability $1-\delta$, one needs $O(d\ln(1/\delta)/(\kappa^2(\epsilon-\epsilon_{\min})^2))$  copies of $|u(t)\rangle$ from Algorithm~\ref{Alg-CVsimulation} and uses of the two mode control-gate $\exp(2i \kappa \nu \hat{p}_k \otimes \hat{p})$ to recover  $\partial S(t, x_a)/\partial x_k$ to precision $\epsilon$.  
\end{lemma}
\begin{proof}
We denote the estimate of $g_{ka}$ using Algorithm~\ref{Alg-CVgradS} by $\tilde{g}_{ka}$. Then combining with Lemma~\ref{lem:gomeserror} 
\begin{align}
   &\Bigg|\frac{\partial S(t, x_a)}{\partial x_k}-|\tilde{g}_{ka}|\Bigg|\leq \Bigg|\frac{\partial S(t, x)}{\partial x_k}-|\tilde{g}_{k}|\Bigg|_{\max} \leq  \|\partial S(t, x)/\partial x_k-|\tilde{g}_{k}|\|_2 \leq  \|\partial S(t, x)/\partial x_k-|g_{k}|\|_2+\||g_{k}|-|\tilde{g}_{k}|\|_2 \nonumber \\
   & \leq O(\tilde{\epsilon})+ \sqrt{d}||g_{k}|-|\tilde{g}_{k}||_{max}=O(\tilde{\epsilon})+\sqrt{d}\zeta=\epsilon_{\min}+\sqrt{d}\zeta=\epsilon, \quad \nu=O(\tilde{\epsilon}^2/d^2),
\end{align}
where $g_{k}=-2i \nu\langle x|\hat{p}_k|u(t)\rangle/\langle x|u(t)\rangle$ and $\zeta$ is the maximum sampling error in estimating $g_{ka}$. 
The total error $\epsilon$ can only be reduced by minimising $\zeta$ error. Suppose one measures the expectation values of the ancilla qumode $\langle \Psi |\hat{x}|\hat{\Psi}\rangle$, $\langle \Psi|\hat{p}|\Psi\rangle$ each to precision $\gamma$ with probability $1-\delta$. Now 
$\gamma$ is related to $\zeta$ using Eq.~\eqref{eq:cvgka}, where
$\zeta=(\gamma/\kappa)(\langle \Psi |\hat{x}|\hat{\Psi}\rangle+\langle \Psi |\hat{p}|\hat{\Psi}\rangle)/\sqrt{\langle \Psi |\hat{x}|\hat{\Psi}\rangle^2+\langle \Psi |\hat{p}|\hat{\Psi}\rangle^2}\sim O(\gamma/\kappa)$. Thus the number of $|\Psi\rangle$ states required is $O(\ln(1/\delta)/\gamma^2)=O(\ln(1/\delta)/(\zeta \kappa)^2)$
where $\zeta=(\epsilon-\epsilon_{\min})/\sqrt{d}$. Each $|\Psi\rangle$ requires the preparation of $O(\phi_{coh}(x^j))$ copies of $|u(t)\rangle$ from Algorithm~\ref{Alg-CVsimulation} and the same number of uses of $\exp(2i \kappa \nu \hat{p}_k \otimes \hat{p})$, so the total number of $|u(t)\rangle$ state preparations required need to be multiplied by $1/\phi_{coh}$. From Lemma~\ref{lem:gomeserror} we see it is sufficient to choose $\nu =O(\epsilon^2_{\min}/d^2)$.
\end{proof}

\subsubsection{Digital protocol for gradient estimation} \label{sec:DVgradient}

We will first make a discrete-variable estimate of $g_{jk}$ from Section~\ref{sec:CVgradient} using the discrete-variable state $|u(t)\rangle_{DV}$ and discrete-variable operator $\hat{P}_k$:
\begin{align}
    \hat{g}_{ka}=-2 i \nu \frac{\langle x_a|\hat{P}_k|u(t)\rangle_{DV}}{\langle x_a|u(t)\rangle_{DV}}.
\end{align}
 Instead of discretising the ancilla coherent state, we can instead find a simpler digital protocol by using a single-qubit ancilla $|0\rangle$. We summarise the protocol in Algorithm~\ref{Alg-DVgradS} with cost given in Lemma~\ref{lem:DVgradient}. 

\begin{algorithm}[H]
		\caption{The digital protocol to estimate $|\hat{g}_{ka}|$ at some given $x=x_a$, see Fig.~\ref{fig:DVgradient}. Output estimates $\partial_k S(t, x_a)$ to precision $\epsilon$ with cost given in  Lemma~\ref{lem:DVgradient}}
		\label{Alg-DVgradS}
\begin{enumerate}
    \item Input:   $|u(t)\rangle_{DV}$ from Algorithm~\ref{Alg-DVsimulation};  copies of single-qubit ancilla state $|0\rangle$; ability to perform unitary operation $\exp(2i \kappa \nu \hat{P}_k \otimes \sigma_x)$, $x_a$;
    \item  Simulate $\exp(2i \kappa\nu \hat{P}_k \otimes \sigma_x)|u(t)\rangle_{DV}|0\rangle$;
    \item Project the non-ancilla state onto $|x_a\rangle \langle x_a|$. This has success probability $\phi(x_a)$;
    \item  If step (3) succeeds, the resulting single-qubit ancilla state is $|\hat{\Psi}\rangle$. Measure $\langle \hat{\Psi}|\sigma_z|\hat{\Psi}\rangle$
    \item Output: $|\hat{g}_{ka}|=\sqrt{2/(1+\langle \hat{\Psi}|\sigma_z|\hat{\Psi}\rangle)-1}$
\end{enumerate}
\end{algorithm}

In the limit when $\kappa \ll 1$, one can apply the projective measurement $|x_a\rangle \langle x_a|$ onto the initial state $|u(t)\rangle |0\rangle$ after applying the unitary operation $\exp(2i \kappa \nu \hat{P}_k \otimes \sigma_x)$: 
\begin{align}
   &  |x_a\rangle \langle x_a|e^{2i \kappa \nu \hat{P}_k \otimes \sigma_x}|u(t)\rangle |0\rangle=|x_a\rangle\langle x_a|u(t)\rangle\left(\mathbf{1}-\kappa \hat{g}_{ka}\sigma_x\right)|0\rangle+O(\kappa^2) \nonumber \\
   & \propto |x_a\rangle |\hat{\Psi}\rangle+O(\kappa^2), \qquad |\hat{\Psi}\rangle= \frac{(\mathbf{1}- \kappa \hat{g}_{ka}\sigma_x)|0\rangle}{\||\Psi \rangle\|}, \qquad \| |\Psi \rangle\|^2=1+|\hat{g}_{jk}|^2, 
\end{align}
where $|\hat{g}_{ka}|^2=\text{Re}(\hat{g}_{ka})^2+\text{Im}(\hat{g}_{ka})^2$. It is simple to verify that the expectation values 
\begin{align}
    \langle \hat{\Psi}|\sigma_x|\hat{\Psi}\rangle=-2\frac{\text{Re}( \kappa\hat{g}_{ka})}{1+\kappa^2|\hat{g}_{ka}|^2}, \qquad  \langle \hat{\Psi}|\sigma_y|\hat{\Psi}\rangle=-2\frac{\text{Im}( \kappa \hat{g}_{ka})}{1+\kappa^2|\hat{g}_{ka}|^2}, \qquad  \langle \hat{\Psi}|\sigma_z|\hat{\Psi}\rangle=\frac{1-\kappa^2|\hat{g}_{ka}|^2}{1+ \kappa^2|\hat{g}_{ka}|^2}. 
\end{align}
Since we only need $|\hat{g}_{ka}|$, it is simple to see 
\begin{align}
    |\hat{g}_{ka}|=\frac{1}{\kappa}\sqrt{\frac{2}{1+\langle \hat{\Psi}|\sigma_z|\hat{\Psi}\rangle}-1}.
\end{align}
To prepare $|\hat{\Psi}\rangle$, we require the application of the hybrid CV-DV gate $\exp(2i \kappa \nu \hat{P}_k \otimes \sigma_x)$, followed by a projective measurement onto $|x_a\rangle$. This projective measurement has success probability
\begin{align}
   \text{Tr}\left((\mathbf{1} \otimes |x_a\rangle \langle x_a|)e^{2i \kappa \nu \hat{P}_k \otimes \sigma_x}(|u(t)\rangle \langle u(t)| \otimes |0\rangle \langle 0|)e^{-2i \kappa \nu \hat{P}_k \otimes \sigma_x}\right)=\phi(x_a),
\end{align}
See Fig.~\ref{fig:DVgradient} for the schematic figure of the protocol to estimate $\hat{g}_{ka}$. The total cost is given in Lemma~\ref{lem:DVgradient}.

\begin{figure}[h] 
\centering
\includegraphics[width=12cm]{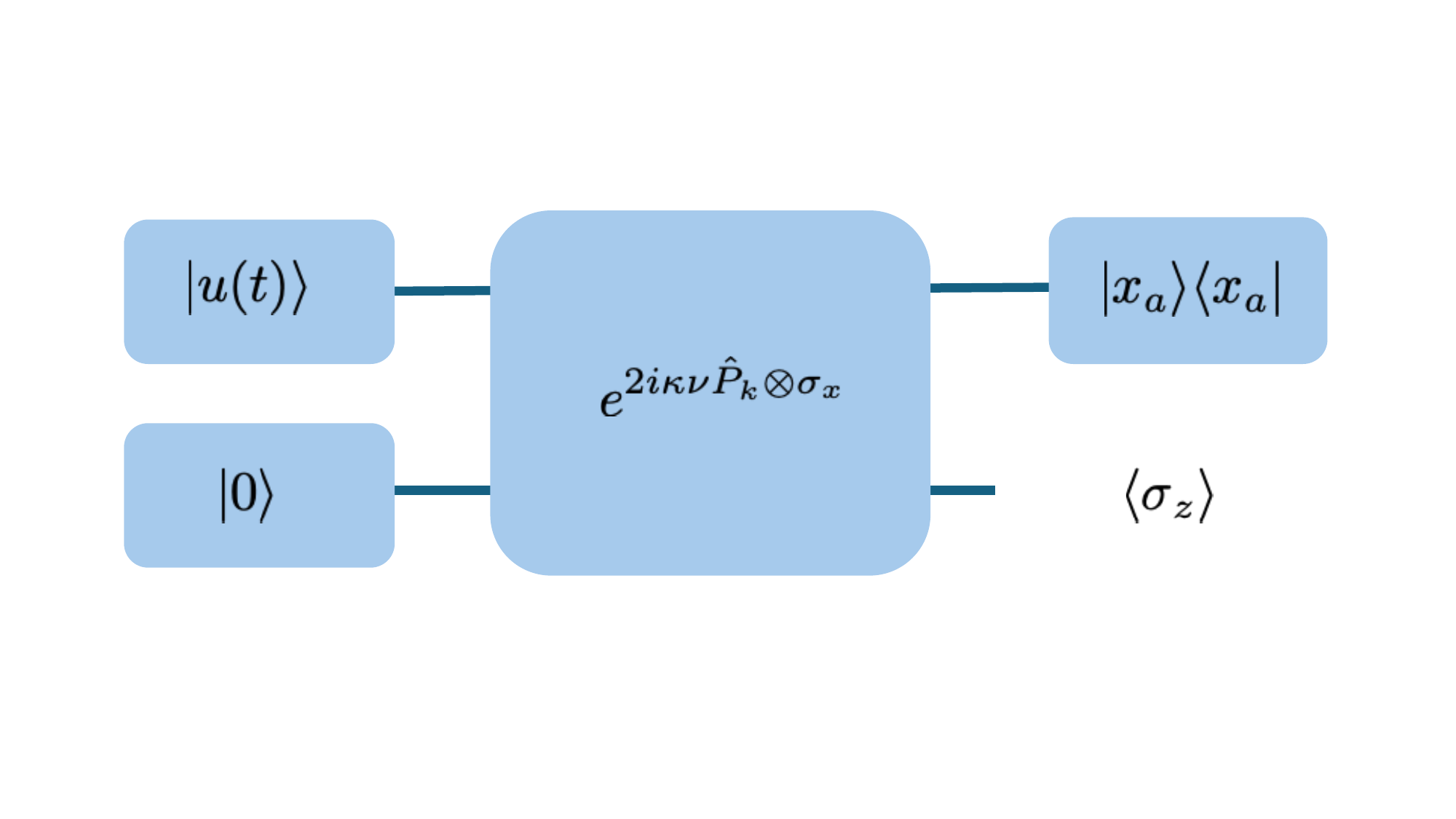} 
%\captionsetup{width=1\textwidth}

\caption{\justifying Schematic diagram of digital quantum protocol for extracting the digital approximation of the gradient $\partial S_{\nu}(t, x=x_a/\partial x_k)$ by  $|\hat{g}_{ka}|$, when given copies of the single-qubit ancilla state $|0\rangle$ and the discrete-variable state $|u(t)\rangle$. Given the initial state $|u(t)\rangle|0\rangle$, we apply the operation $\exp(2i \kappa \nu \hat{P}_k \otimes \sigma_x)$, where $\hat{P}_k$ acts on the $|u(t)\rangle$ register and $\sigma_x$ acts on the qubit ancilla and $\kappa \ll 1$. By projecting the resulting first register onto the desired position $|x_a\rangle \langle x_a|$, the resulting qubit state can be used to estimate $|\hat{g}_{jk}|$ by measuring $\langle \sigma_z\rangle$ of the resulting ancilla qubit.}\label{fig:DVgradient}
\end{figure}

\begin{lemma} \label{lem:DVgradient}
    Let the minimum error to recover $\partial S(t, x_a)/\partial x_k$ be $\epsilon_{\min}$. To estimate $\partial S(t, x_a)/\partial x_k$ to precision $\epsilon$ requires $O(d\ln(1/\delta)/(\kappa^2(\epsilon-\epsilon_{\min})^2))$ copies of $|u(t)\rangle$ from Algorithm~\ref{Alg-DVsimulation}, and the same uses of the qubit-control gate $\exp(2i \epsilon \nu \hat{P}_k \otimes \sigma_x)$, where $\hat{P_k}$ acts on a system of $O(d\ln(N_x))$ qubits. In  Algorithm~\ref{Alg-DVsimulation} it is sufficient to choose $\nu=O(\epsilon^2_{\min}/d^2)$ where $\epsilon_{\min}=d^{3/2}/N_x$.
\end{lemma}
\begin{proof}
Here the results are very similar to the analog case, except we need to take into account the discretisation error. Thus following the same reasoning as the analog case using Lemma~\ref{lem:gomeserror} and Lemma~\ref{error-dis} 
\begin{align}
   &\Bigg|\frac{\partial S(t, x_a)}{\partial x_k}-|\tilde{\hat{g}}_{ka}|\Bigg|\leq \Bigg|\frac{\partial S(t, x)}{\partial x_k}-|\tilde{\hat{g}}_{k}|\Bigg|_{\max} 
    \leq O(\tilde{\epsilon})+ \sqrt{d}||\hat{g}_{k}|-|\tilde{\hat{g}}_{k}||_{max}=O(\tilde{\epsilon})+\sqrt{d}\zeta=\epsilon_{\min}+\sqrt{d}\zeta=\epsilon, \nonumber \\
    & \nu=O(\tilde{\epsilon}^2/d^2), \quad 1/N_x=O(\tilde{\epsilon}/d^{3/2}),
\end{align}
where $\hat{g}_{ka}$ differs from  $g_{ka}=\partial S_{\nu}(t, x_a)/\partial x_k$ also by discretisation errors analysed in Lemma~\ref{error-dis}. The total error is $\epsilon=\epsilon_{\min}+\zeta$ where $\zeta$ is the precision of estimating $\hat{g}_{jk}$ using Algorithm~\ref{Alg-DVgradS}. Suppose we measure the expectation values of the ancilla qubit $\langle \hat{\Psi}|\sigma_z|\hat{\Psi}\rangle$ to  precision $\eta$ with probability $1-\delta$. Then $\zeta \sim \eta/(\kappa \hat{g}_{ka}(1+\langle \hat{\Psi}|\sigma_z|\hat{\Psi}\rangle)^2)\sim O(\eta/\kappa)$. The number of $|\Psi\rangle$ states required is therefore $O(\ln(1/\delta)/\eta^2)=O(\ln(1/\delta)/(\kappa^2\zeta^2)$, where $\zeta=(\epsilon-\epsilon_{\min})/\sqrt{d}$. Each $|\hat{\Psi}\rangle$ requires the preparation of $O(\phi(x^j))$ copies of $|u(t)\rangle$ from Algorithm~\ref{Alg-DVsimulation} and the same number of uses $\exp(2i\nu \hat{P}_k \otimes \sigma_x)$. 
\end{proof}

\subsection{Estimating the minimum value of the solution} \label{sec:maximalSCV}

For the following algorithms, we will use the following Assumption~\ref{assump:S}.  
\begin{assumption}\label{assump:S}
For the function $S_{\nu}(t,x)$ at some constant $t$, we will assume the following properties: 
   \begin{enumerate}
    \item $x^*_{\nu}$ is the unique global minimiser of $S_{\nu}(t, x)$; 
    \item $x^*_{\nu}$ is an interior point in the domain; 
    \item The Hessian $H_{S_{\nu}}(t, x^*_{\nu})=\nabla^2S_{\nu}(t, x^*_{\nu})$ is positive definite (so we have a minimum at that point) 
    \item $\|\mathbf{u}(t)\|_{CV}^2=\int e^{-S_{\nu}(t, x)/\nu} dx<\infty$ since we assume $|u(t)\rangle_{CV}$ is a physical state. No extra assumption is needed in the discrete-variable case since the normalisation constant  is always bounded. 
\end{enumerate}
\end{assumption}

We begin by defining the minimum value of $S_{\nu}(t, x)$ at some given time $t$: 
\begin{align}
    S_{\nu, \min}=\inf_{x} S_{\nu}(t, x), 
\end{align}
where we suppress the $t$ in $S_{\nu, \min}$ to simplify notation. Then using Assumption~\ref{assump:S} for $S_{\nu}$, we can show how $S_{\nu, \min}$ can be estimated using the normalisation constant $\|\mathbf{u}(t)\|^2$ for the state $|u(t)\rangle$. We summarise the analog protocol to estimate $S_{\min}$ in Algorithm~\ref{Alg-CVSmin} and find the cost in Lemma~\ref{lem:cvsmin}.

\subsubsection{Analog protocol to estimate $S_{\min}$}
\begin{algorithm}[H]
		\caption{The analog protocol to estimate $S_{\min}$. Output estimates $S_{\min}$ to precision $\epsilon_{S^*}$ with cost given in Lemma~\ref{lem:cvsmin}.}
		\label{Alg-CVSmin}
\begin{enumerate}
    \item Input: $\|\mathbf{u}(t)\|_{CV}^2$ from Algorithm~\ref{Alg-CVsimulation}; 
    \item Output: $-\nu \ln \|\mathbf{u}(t)\|_{CV}^2$.
\end{enumerate}
\end{algorithm}

 In the small $\nu \leq O(1/d) \ll 1$ regime, it is possible to show (see Appendix~\ref{app:cvsmin}) that $S_{\nu, \min}$ can be approximated using $-\nu \ln \|\mathbf{u}(t)\|^2$ such that 
\begin{align} \label{eq:smaxapprox} 
    |S_{\nu, \min}+\nu \ln \|\mathbf{u}(t)\|^2| \lesssim \tilde{O}(\nu d),
\end{align}
where $\tilde{O}$ contains log factors of $\nu$. From Section~\ref{sec:quantumsimulation} and Fig.~\ref{fig: Normalisation}, one sees that there is a simple protocol to estimate the normalisation constant $\|\mathbf{u}\|^2$ of $|u(t)\rangle$ as a byproduct of preparing $|u(t)\rangle$, and this therefore also becomes a protocol to estimate $S_{\nu, \min}$. Combining Eq.~\eqref{eq:smaxapprox} with Lemma~\ref{lem:gomeserror},  in order to use the normalisation constant $\|\mathbf{u}(t)\|^2$ to estimate the minimum of the viscosity solution $S_{\min}=\inf_{x} S(t,x)$ to precision $\epsilon_{\min}$, the minimum possible error in the scheme is:  
\begin{align}
    |S_{\min}+v\ln \|\mathbf{u}(t)\|^2| \leq |S_{\min}-S_{\nu, \min}|+|S_{\nu,\min}+\nu\ln \|\mathbf{u}(t)\|^2|\lesssim O(\tilde{\epsilon})+ \tilde{O}(\nu d)=\epsilon_{\min},  
\end{align}
where $\nu=O(\tilde{\epsilon}^2/d^2)$.
Thus the dominant error still comes from $|S_{\min}-S_{\nu, \min}|$, so we can write $\epsilon_{\min} \sim \sqrt{\nu} d$. This means that even when the precision of estimation for $\|\mathbf{u}(t)\|^2$ is perfect (i.e. an infinite number of measurements in the protocol in Fig.~\ref{fig: Normalisation}), there is still a minimum error $\epsilon_{
\min}$ in estimating $S_{\min}$ dependent on $\nu$, and $\epsilon_{\min} \rightarrow 0$ only when $\nu \rightarrow 0$.  \\

However, we only have the above minimum error when the normalisation constant itself is perfectly known. In general, the normalisation constant itself is not known perfectly due to a limited (finite) number of measurements allowed in the quantum protocol. Below we give the cost in number of measurements of the final ancilla state in Fig.~\ref{fig: Normalisation} to obtain the estimate $\|\tilde{\mathbf{u}}(t)\|^2$, so $S_{\min}$ can be obtained to precision $\epsilon_S$, which is a combination of $\epsilon_{\min}$ and the sampling error from the quantum protocol. 

\begin{lemma} \label{lem:cvsmin}
   The minimum precision in estimating $S_{\min}$ using the normalisation constant of $S_{\nu}$ is $\epsilon_{\min}$. Let $\epsilon_{S^*}$ be the total error in estimating $S_{\min}$ when using Algorithm~\ref{Alg-CVsimulation} to estimate the normalisation constant $\|\mathbf{u}(t)\|^2$. Then the number of preparations of state $|v(t)\rangle$ in Step (2) of  Algorithm~\ref{Alg-CVsimulation} required to estimate the normalisation constant is 
   \begin{align}
    O\left(\frac{\nu^2}{\|\mathbf{u}(t)\|^4}\ln\left(\frac{1}{\delta}\right)\frac{1}{(\epsilon_{S^*}-\epsilon_{\min})^2}\right), 
\end{align}
where it is sufficient to choose $\nu=O(\epsilon_{\min}^2/d^2)$ for some fixed $\epsilon_{\min}$. 
   \end{lemma}
\begin{proof}
See the proof of Eq.~\eqref{eq:smaxapprox} in Appendix~\ref{app:cvsmin}, where in the proof of Eq.~\eqref{eq:smaxapprox} we used $\nu \leq O(1/d)$. In order to estimate the minimum viscosity solution $S_{\min}$ to precision $\epsilon_S$ using the estimated $\|\tilde{\mathbf{u}}(t)\|^2$, we combine Eq.~\eqref{eq:smaxapprox} and Lemma~\ref{lem:gomeserror} and see that it is sufficient to impose 
   \begin{align} \label{eq:sminerror}
      &  |S_{\min}+\nu \ln \|\tilde{\mathbf{u}}(t)\|^2| \leq |S_{\min}-S_{\nu, \min}|+|S_{\nu,\min}+\nu \ln \|\mathbf{u}(t)\|^2|+|\nu \ln \|\mathbf{u}(t)\|^2-\nu \ln \|\tilde{\mathbf{u}}(t)\|^2| \nonumber \\
      & \leq O(\tilde{\epsilon})+ \nu d\ln(2\pi \nu/\det H_{S_{\nu}}(x^*))/2+\tilde{O}_{hot}(\nu)+\zeta \lesssim  O(\tilde{\epsilon})+\tilde{O}(\nu d)+\zeta=\epsilon_{\min}+\zeta =\epsilon_{S^*},
   \end{align}
   where $\nu=O(\tilde{\epsilon}^2/d^2)$ and only the last error $\epsilon$ can be improved by increasing sampling in the quantum protocol in Fig.~\ref{fig: Normalisation}, and $\epsilon_{\min}(\nu)$ is the infimum error that can be achieved.  \\

One can relate $\zeta$ to the sampling cost in the following way. Let $\epsilon_{norm}$ denote the precision to which one can measure the normalisation $\|\mathbf{u}(t)\|^2$. Let $1-\delta$ be the success probability, then the total measurement cost (which is equivalent to the number of repetitions of the quantum simulation with respect to $\mathbf{H}(t)$, since the normalisation is estimated as a byproduct -- it is the probability of measuring the final ancilla state in $\xi>0$) in the protocol in Fig.~\ref{fig: Normalisation} is $O(\ln(1/\delta)/\epsilon^2_{norm})$. Then we can use 
\begin{align}
    |\nu \ln \|\mathbf{u}(t)\|^2-\nu \ln \|\tilde{\mathbf{u}}(t)\|^2| \leq \zeta \implies \zeta \sim \nu \epsilon_{norm}/\|\mathbf{u}(t)\|^2,
\end{align}
so the sampling cost $O(\ln(1/\delta)/\epsilon^2_{norm}) \sim O(\nu^2\ln(1/\delta)/(\zeta^2\|\mathbf{u}(t)\|^4)$. From Eq.~\eqref{eq:sminerror} it is sufficient to choose $\nu=O(\epsilon^2_{\min}/d^2)$ and $\zeta=\epsilon_{S^*}-\epsilon_{\min}(\nu)$, so the quantum sampling cost $O(\nu^2\ln(1/\delta)/(\zeta^2\|\mathbf{u}(t)\|^4)$  can be rewritten in terms of $\epsilon_{\min}(\nu)$ and $\epsilon_{S^*}$. 
\end{proof}

\subsubsection{Digital protocol to estimate $S_{\min}$}
\begin{algorithm}[H]
		\caption{The digital protocol to estimate $S_{\min}$. Output estimates $S_{\min}$ to precision $\epsilon_{S^*}$ with cost given in Lemma~\ref{lem:DVsmin}.}
		\label{Alg-DVSmin}
\begin{enumerate}
    \item Input: $\|\mathbf{u}(t)\|_{DV}^2$ from Algorithm~\ref{Alg-DVsimulation};
    \item Output: $-\nu \ln \|\mathbf{u}(t)\|_{DV}^2$.
\end{enumerate}
\end{algorithm}

In the digital protocol, one can proceed similarly, where the minimal viscosity solution of the Hamilton-Jacobi equation can also be approximated using only the normalisation constant. In the digital case, give the  discrete mesh points  $x_j=j/N_x$, and define the minimum of $S_{\nu}$ by 
\begin{align}
    S_{\nu, min}(t)=\min_{j} S_{\nu}(t, x_j).
\end{align}
One can use a similar bound to Eq.~\eqref{eq:smaxapprox} to estimate $S_{\nu, \min}$, except we can derive it in a much simpler way 
(similar to the LogSumExp (LSE) function $\ln \sum_j \exp(y_j)$, used widely in machine learning to construct a smooth approximation to a maximum function). It is simple to see the following (using $S_{\nu, \Delta, \min}(t)<S_{\nu, \Delta}(t, x_j)$ for any $x_j$ that is not the minimum point) and Lemma~\ref{lem:gomeserror} and Lemma~\ref{error-dis}:
\begin{align} \label{eq:logsum}
   & \bigg|\nu \ln \left(\|\mathbf{u}(t)\|_{DV}^2\right)+S_{\min}\bigg| \leq \bigg|\nu \ln \left(\|\mathbf{u}(t)\|_{DV}^2\right)+S_{\nu,\Delta, \min}\bigg|+\bigg|S_{\nu, \Delta, \min}-S_{\min}\bigg| \leq \bigg|\nu \ln \left(\sum_{j}^{N_x^d}e^{-\frac{S_{\nu, \Delta}(t, x_j)}{\nu}}\right)+S_{\nu,\Delta, \min}\bigg|+\tilde{\epsilon}\nonumber \\
   &=\bigg|\nu \ln\left(e^{-S_{\nu,\Delta, \min}(t)/\nu}\sum_{j}^{N_x^d}e^{\frac{-S_{\nu, \Delta}(t, x_j)+S_{\nu,\Delta, \min}(t)}{\nu}}\right)+S_{\nu,\Delta, \min}\bigg| +\tilde{\epsilon}=\nu \ln \left(\sum_{j}^{N_x^d}e^{\frac{-(S_{\nu, \Delta}(t, x_j)-S_{\nu,\Delta, \min}(t))}{\nu}}\right) +\tilde{\epsilon} \nonumber \\
   &\leq \nu \ln (N_x^d)+\tilde{\epsilon}=\nu d \ln N_x+\tilde{\epsilon}=\epsilon_{\min}, \qquad \nu=O(\tilde{\epsilon}^2/d^2), \qquad 1/N_x=O(\sqrt{\tilde{\epsilon}/d}),
\end{align}
where $\epsilon_{\min}$ is the minimum error in estimating $S_{\min}$ even if the normalisation constant is perfectly known, just like in the case of the analog protocol.

\begin{lemma} \label{lem:DVsmin}
   Let minimum precision in estimating $S_{\min}$ using the normalisation constant from $S_{\nu}$ be $\epsilon_{\min}$ when one uses $\ln N_x$ qubits to represent $|u(t)\rangle_{DV}$. Let $\epsilon_{S^*}$ be the total error in estimating $S_{\min}$ when we are using Algorithm~\ref{Alg-DVsimulation} to estimate the normalisation constant $\|\mathbf{u}(t)\|_{DV}^2$. The number of preparations of state $|v(t)\rangle$ in Step (2) of Algorithm~\ref{Alg-2} to estimate the normalisation constant is 
   \begin{align}
    O\left(\frac{\nu^2}{\|\mathbf{u}(t)\|^4}\ln\left(\frac{1}{\delta}\right)\frac{1}{(\epsilon_{S^*}-\epsilon_{\min})^2}\right), 
\end{align}
where the constant $\nu$ used in the quantum simulation algorithm can be chosen $\nu=1/N_x^4$ and the resulting minimum error is  $\epsilon_{\min}=\sqrt{\nu}d+\nu d\ln(1/\nu)/4$. 
\end{lemma}
\begin{proof}
Here Algorithm~\ref{Alg-DVsimulation} can be used to estimate the normalisation constant. The proof then proceeds in the same way as Lemma~\ref{lem:cvsmin}, except now $\epsilon_{
min}$ is slightly modified, see Eq.~\eqref{eq:logsum}. Using Eq.~\eqref{eq:logsum},  one can for example choose $N_x^2=d/\tilde{\epsilon}$ and $\nu=\tilde{\epsilon}^2/d^2$, which implies $\nu=1/N_x^4$. Then the minimum error can be rewritten $\epsilon_{\min}=\sqrt{\nu}d+\nu d\ln(1/\nu)/4=d/N_x^2+d\ln(N_x)/N_x^4 \sim d/N_x^2$. 
\end{proof}

%If one were to perform a search over all the computed $S_{\nu}$, the search problem would be in a space of size $N_x^d$, thus the search itself would be highly inefficient, of order $O(N_x^d)$, unless $S_{\nu, \min}$ is unimodal, which is not guaranteed, for example if we started with non-convex initial conditions in the Hamilton-Jacobi equation. 

 \subsection{Estimating value of known $f(x)$ at the minimal point} \label{sec:knownfunctionmax}
 
 Suppose there is a known function $f: X \rightarrow \mathbb{R}$ and one wants to estimate the value of this function at the minimal  point $x=x^*$ of $S$ at a particular $t$ defined by 
\begin{align}
    x^*=\text{argmin}_x S(t, x). 
\end{align}
However, one only has direct access to $S_{\nu}(t,x)$, and we define its minimum to occur at $x=x^*_{\nu}$
\begin{align}
    x^*_{\nu}=\text{argmin}_x S_{\nu}(t, x).
\end{align}
We can bound $\|x^*-x^*_{\nu}\|$ if for example by assuming a strongly convex initial condition $S_0(x)$ function, which for convex $\mathcal{H}(\nabla S, x)$ as we consider, implies convexity in the solution $S_{\nu}(t,x)$. Furthermore, if  assuming that $f(x)$ is a Lipschitz continuous function, then one can also bound $|f(x^*)-f(x^*_{\nu})|$.

\begin{lemma} \label{lem:fdiff}
    Let $S_0(x)$ be a smooth, $\mu$-strongly convex function, and  $f(x)$ be a Lipschitz continuous function with Lipschitz constant $K$. Then 
    \begin{align}
        |f(x^*)-f(x^*_{\nu})| \leq 2K\epsilon/\sqrt{\mu}, \qquad \nu=O(\epsilon^4/d^2).
    \end{align}
\end{lemma}
\begin{proof}
    Suppose $S_{\nu}(t,x)$ is a  $\mu_t$-strongly convex function at time $t$, so for any $x \in \mathbb{R}^d$, $S_{\nu}(x) \geq S_{\nu}(x^*_{\nu})+\mu_t\|x-x^*_{\nu}\|^2/2$, since $\nabla S_{\nu}(x^*_{\nu})=0$. Inserting $x=x^*$, one gets the inequality $S_{\nu}(x^*)-S_{\nu}(x^*_{\nu}) \geq \mu_t\|x^*-x^*_{\nu}\|^2/2$. From Lemma~\ref{lem:gomeserror}, one knows that $-\epsilon \leq S_{\nu}(x)-S(x) \leq \epsilon$ where $\nu=O(\epsilon^2/d^2)$, which implies the inequality chain $S_{\nu}(x^*)-\epsilon \leq S(x^*) \leq S(x^*_{\nu}) \leq S_{\nu}(x^*_{\nu})+\epsilon$, so $S_{\nu}(x^*) \leq S_{\nu}(x^*_{\nu})+2\epsilon$. Together this implies $\|x^*-x^*_{\nu}\|\leq 2\sqrt{\epsilon/\mu_t}$.  Using this with the Lipschitz continuity of $f(x)$, this implies 
    \begin{align}
      |f(x^*)-f(x^*_{\nu})|\leq K \|x^*-x^*_{\nu}\| \leq 2K \epsilon/\sqrt{\mu_t}, \qquad \nu=O(\epsilon^4/d^2).
    \end{align}
    It is possible to show, when $h$ is convex,  that $\mu_t \geq \mu$, where the initial condition $S_{\nu}(0,x)=S_0(x)$ is a $\mu$-strongly convex function. 
\end{proof}
For our following protocols to estimate $f(x^*)$, in addition to use  Assumption~\ref{assump:S} on $S_{\nu}(t,x)$, we also require the additional assumption of $S_0(x)$ being $\mu$-strongly convex (equivalent to $H_S(0,x^*_{\nu}) \geq \mu I)$. 

\subsubsection{Analog protocol to compute $f(x)$ at minimal point}
\label{sec:cvfmin}
We summarise the analog protocol to estimate $f(x^*)$ in Algorithm~\ref{Alg-CVfmin}, with cost given in Lemma~\ref{lem:cvfmax}. 
\begin{algorithm}[H]
		\caption{The analog protocol to estimate $f(x^*)$. Output estimates $f(x^*)$ to precision $\epsilon_f$ with cost given in Lemma~\ref{lem:cvfmax}.}
		\label{Alg-CVfmin}
\begin{enumerate}
    \item Input: Access to $f(\hat{x})$ as a sum of Hermitian operators and  $|u(t)\rangle_{CV}$ from Algorithm~\ref{Alg-CVsimulation}; 
    \item Output: $\langle u(t)|f(\hat{x})|u(t)\rangle_{CV}$.
\end{enumerate}
\end{algorithm}
 
In this section, we will show how by preparing $|u(t)\rangle$ and by promoting each $x_i \rightarrow \hat{x}_i$ in the function $f(x=(x_1, \cdots, x_d))$ to the operator form using the shorthand $f(\hat{x})$, we can estimate $f(x^*_{\nu})$ to high precision by measuring the expectation value $\langle u(t)|f(\hat{x})|u(t)\rangle$. Schematically, when $\nu$ is very small so one can make use of Laplace's method, we can write (see Appendix~\ref{app:fmin})
\begin{align} \label{eq:cvfstar}
    |\langle f(\hat{x}) \rangle -f(x^*_{\nu})|=\Bigg|\frac{\int f(x)e^{-S_{\nu}(t, x)/\nu}dx}{\int e^{-S_{\nu}(t,x)/\nu}dx}- f(x^*_{\nu})\Bigg| \leq O(\nu), \qquad \langle f(\hat{x}) \rangle =\langle u(t)|f(\hat{x})|u(t)\rangle,
\end{align}
plus higher order terms in $\nu$. 
Since $f(x)$ can always be approximated by sums of polynomials, the expectation value $\langle f(\hat{x}) \rangle$ can always be measured with homodyne measurements to obtain moments of $\hat{x}$. Denoting the estimate of this expectation by $\langle \tilde{f}(\hat{x})\rangle$, then one can only use the quantum protocol to reduce $|\langle f(\hat{x}) \rangle -\langle \tilde{f}(\hat{x}) \rangle | \leq \epsilon$ with more accurate estimation by quantum sampling. So, using Lemma~\ref{lem:fdiff} and Eq.~\eqref{eq:cvfstar}, the total error in obtaining $f(x^*)$ is $\epsilon_f$
\begin{align} \label{eq:fineq}
  &  |f(x^*)-\langle \tilde{f}(\hat{x})\rangle| \leq |f(x^*)-f(x^*_{\nu})|+|f(x^*_{\nu})-\langle f(\hat{x}) \rangle |+ |\langle f(\hat{x}) \rangle -\langle \tilde{f}(\hat{x})\rangle| \nonumber \\
  & \leq O(\epsilon)+O(d\nu)+\zeta=\epsilon_{\min}+\zeta=\epsilon_f, \qquad \nu=O(\epsilon^4/d^2),
\end{align}
 $\epsilon_{\min}$ is the minimum error that cannot be improved by increased quantum sampling of the expectation value, and can only decrease with a smaller choice of $\nu$. The dominant error in $\epsilon_{\min}$ comes from the first term $|f(x^*)-f(x^*_{\nu})|$, so only keep that. 

\begin{lemma} \label{lem:cvfmax}
 Given $f(x)=\sum_{i=1}^L\alpha_i g_i(x)$ which can be written as a sum of $L$ terms where one can access $g_i(\hat{x})$ as hermitian operators, where we can measure the expectation value of each $g_i(\hat{x})$ with repetitions of some measurement $\hat{\Pi}_i$. The minimum precision in estimating $f(x^*)$ using $\int f(x) \exp(-S_{\nu}/\nu) dx/\int \exp(-S_{\nu}/\nu)dx$ is $\epsilon_{\min}$. Then to estimate $f(x^*)$ to precision $\epsilon_f$ one can use Algorithm~\ref{Alg-CVfmin}. Here it is sufficient to prepare $O(L\max_i \langle g_i^2(\hat{x})\rangle/(\epsilon_f -\epsilon_{\min})^2)$ copies of $|u(t)\rangle$ (using Algorithm~\ref{Alg-CVsimulation} with $\nu=O(\epsilon_{\min}^4/d^2)$) to estimate $\langle f(\hat{x})\rangle$. 
\end{lemma}

\begin{proof}
Firstly, it is possible to show the following bound 
 \begin{align} \label{eq:fmaxbound}
     | \langle f(\hat{x}) \rangle -f(x^*_{\nu})| \leq B \nu, \qquad B=\frac{\nabla^2f(x^*_{\nu})}{2}|\text{Tr}(\nabla^2S_{\nu}(t, x^*_{\nu})^{-1})|, \qquad \langle f(\hat{x}) \rangle =\langle u(t)|f(\hat{x})|u(t)\rangle,
 \end{align}
 including higher order terms in $\nu$ (which we ignore for $\nu \ll 1$), see Appendix~\ref{app:fmin}. From the definition of $\mu_t$-strong convexity of $S_{\nu}(t,x)$ so that $\nabla^2S_{\nu}(t,x^*_{\nu}) \geq \mu_t I$, and $\mu \geq \mu_t$ from Lemma~\ref{lem:fdiff}, then $\text{Tr}(\nabla^2S_{\nu}(t, x^*_{\nu})^{-1}) \leq d/\mu$, so $B \leq \nabla^2 f(x^*_{\nu})d/(2\mu)$. From Lemma~\ref{lem:fdiff}, we see that using $\nu=O(\epsilon^4/d^2)$, the dominant contribution to $\epsilon_{\min}$ comes from $|f(x^*)-f(x^*_{\nu})|$, so it's sufficient to choose $\nu=O(\epsilon_{\min}^4/d^2)$, to ensure that the minimum error (even with the perfect estimation of the quantum expectation value) does not exceed $\epsilon_{\min}$.  For the quantum protocol to estimate $\langle u(t)|f(\hat{x})|u(t)\rangle$ to precision $\epsilon$, and using $\alpha_i=O(1)$, it is sufficient to have $O(L\max_i \langle g_i^2(\hat{x})\rangle/\zeta^2)$ copies of $|u(t)\rangle$, where $\zeta=\epsilon_f-\epsilon_{\min}$.  
\end{proof}

\subsubsection{Digital protocol to compute $f(x)$ at minimal point}
\label{sec:dvfmin}

\begin{algorithm}[H]
		\caption{The digital protocol to estimate $f(x^*)$. Output estimates $f(x^*)$ to precision $\epsilon_f$ with cost given in Lemma~\ref{lemma:dvfmax}.}
		\label{Alg-DVfmin}
\begin{enumerate}
    \item Input: Access to $f(\hat{X})$ as a sum of Hermitian operators and  $|u(t)\rangle_{DV}$ from Algorithm~\ref{Alg-DVsimulation}; 
    \item Output: $\langle u(t)|f(\hat{X})|u(t)\rangle_{DV}$.
\end{enumerate}
\end{algorithm}
 One can proceed, like in Section~\ref{sec:cvfmin}, but now we discretise $x \in \mathbb{R}^d$ by decomposing it into a grid of size $N_x$ in each dimension. Here we can take the expectation value of $f(\hat{X})$ with respect to the now qubit-based quantum state $|u(t)\rangle=\sum_{j=1}^{N_x^d}u(t, x_j=j/N_x^d)|j\rangle$, where the discrete position operator $\hat{X}$ is represented by a $N_x^d \times N_x^d$ diagonal matrix, with elements $\text{diag}(-1/2N_x, \cdots, 1/2N_x)$ in the one-dimensional case. Then it can be seen  that, since the expectation value of $f(\hat{X})$ with respect to the digital quantum state $|u(t)\rangle$ -- up to discretisation error -- can be written as the expectation value of $f(\hat{x})$ with respect to the continuous-variable quantum state $|u(t)\rangle_{CV}$, thus it is also dominated by $f(x)$ when $x$ is the minimal point of $S_{\nu}$ 
\begin{align} \label{eq:Fint}
    \langle f(\hat{X})\rangle =\frac{\sum_{j=1}^{N_x^d} f(x_j)e^{-S_{\nu}(t, x_j)/\nu}}{\sum_{j=1}^{N_x^d} e^{-S_{\nu}(t,x_j)/\nu}} \approx \frac{\int f(x)e^{-S_{\nu}(t, x)/\nu}dx}{\int e^{-S_{\nu}(t,x)/\nu}dx} \approx f(x^*_{\nu}), \qquad  \langle f(\hat{X})\rangle =\langle u(t)|f(\hat{X})|u(t)\rangle,
\end{align}
where we choose the discretisation in a box $[-1/2, 1/2]^d$ so the volume of the hypercube is $1$ for simplicity. Then $\epsilon_{\min}$ from Eq.~\eqref{eq:fineq} needs to be augmented by an extra factor that is the error in the discretisation of the integral (e.g. using the midpoint rule) in Eq.~\eqref{eq:Fint} 
\begin{align}
    |\langle f(\hat{x})\rangle-\langle f(\hat{X})\rangle|\lesssim O(1/(\nu N_x^2)),
\end{align}
so now the minimum error is $\epsilon_{\min}\sim O(\tilde{\epsilon})+O(1/(\nu N_x^2))$ for $\nu=O(\tilde{\epsilon}^4/d^2)$. The quantum protocol in estimating $\langle f(\hat{X})\rangle$ then has the cost given below.

\begin{lemma} \label{lemma:dvfmax}
Given a known function $f(x)=\sum_{i=1}^L\alpha_i g_i(x)$ which can be written as a sum of $L$ polynomial terms. Then to estimate $f(x^*)$ to precision $\epsilon_f$ with the minimal accepted precision $\epsilon_{\min}$, it is sufficient to prepare $O(L\max_i \langle g_i^2(\hat{x})\rangle/(\epsilon_f -\epsilon_{\min})^2)$ copies of $|u(t)\rangle$ to estimate $\langle f(\hat{X})\rangle$. The constant $\nu$ used in Algorithm~\ref{Alg-DVsimulation} to prepare $|u(t)\rangle$ can be chosen to be $\nu=(\epsilon_{\min}-O(\sqrt{d/N_x}))/N_x$.
\end{lemma}
\begin{proof}
   It can be shown (see Appendix~\ref{app:dvfmax}) that
   \begin{align}
        |\langle f(\hat{X})\rangle-\langle f(\hat{x})\rangle|\lesssim O(1/(\nu N_x^2)).
    \end{align}
    Thus the total error $\epsilon_f$ can be written as a sum
    \begin{align}
        &|f(x^*)-\langle \tilde{f}(\hat{X})\rangle|\leq |f(x^*)-f(x^*_{\nu})|+|f(x^*_{\nu})-\langle f(\hat{x})\rangle|+|\langle f(\hat{x})\rangle-\langle f(\hat{X})\rangle|+|\langle f(\hat{X})\rangle-\langle\tilde{f}(\hat{X})\rangle| \nonumber \\
        & \leq O(\tilde{\epsilon})+O(d\nu)+O(1/(\nu N_x^2))+\epsilon=\epsilon_{\min}+\epsilon=\epsilon_f, \qquad \nu=O(\tilde{\epsilon}^4/d^2).
    \end{align}
    Here the expectation value should approximately agree with the analog case, except one now only needs to estimate the discretisation error.
    So, combining with the results in  Lemma~\ref{lem:cvfmax} and only including the dominant error in $\epsilon_{\min}$, the total error becomes 
    \begin{align}
          &  |\langle u(t)|f(\hat{X})|u(t)\rangle-f(x^*)|\leq  |\langle u(t)|f(\hat{X})|u(t)\rangle-\langle u(t)|f(\hat{x})|u(t)\rangle_{CV}|+  |f(x^*)-\langle u(t)|f(\hat{x})|u(t)\rangle_{CV}| \nonumber \\
          & \lesssim O(1/(\nu N_x^2))+O(\tilde{\epsilon})=\epsilon_{\min}, \qquad \nu=O(\tilde{\epsilon}^4/d^2).
    \end{align}
    where the choice $\nu=\tilde{\epsilon}^4/d^2$ gives $\epsilon_{\min}=O(\tilde{\epsilon}+d^2/(N_x^2 \tilde{\epsilon}^4))$. Suppose for simplicity one chooses $\nu =\tilde{\epsilon}^4/d^2= \beta/N^2_x$ for a small $\beta \ll 1$. Then $\epsilon_{\min}=O(\tilde{\epsilon})+\beta$ for $\beta \ll 1$, where $\tilde{\epsilon}=\sqrt{d/N_x}$, so $\nu=\beta/N_x^2=(\epsilon_{\min}-O(\sqrt{d/N_x}))/N_x$.
    
%Combined with Lemma~\ref{lem:gomeserror}, this implies $\beta/N_x^2 \geq h^{1/2}$. When using explicit discretisation methods for example, one needs $h<1/(N_x^2d)$ for stability, which is consistent with the above choice when $N_x^2>d \beta^2$, which can be easily satisfied. This means one can rewrite $\epsilon_{\min}=\beta+B\zeta/N_x^2+L\zeta/N_x^{1/2}$.   These results can easily be extended using higher-order methods for the approximation of the integral beyond the midpoint rule and assuming higher regularity of $f(x)$, and this would lead to a smaller $\epsilon_{\min}$.
\end{proof}

\section{Summary and outlook}

In this work, we have established a unified framework for quantum simulation of viscosity solutions to nonlinear Hamilton–Jacobi  equations with convex Hamiltonians, by combining the entropy penalisation method of Gomes and Valdinoci with quantum simulation of linear parabolic dynamics. This framework generalises the classical Cole–Hopf transformation, enabling a linear approximation of nonlinear viscous Hamilton–Jacobi equations that is valid for more general convex nonlinearities (quadratic and beyond) and arbitrary evolution times. This construction provides a relatively rare example of a nonlinear PDE class, {\it with strong nonlinearity, small dissipation and admitting a global-in-time} linear quantum formulation. This extends beyond perturbative truncations such as Carleman embeddings or other linear approximation methods that are in general not valid globally-in-time,  with strong nonlinearities and small dissipation.\\

Our results demonstrate that the entropy-penalised linear methods retains the essential features of the nonlinear dynamics while remaining amenable to efficient quantum evolution. In particular, the viscosity regularisation ensures convergence to physically meaningful weak solutions even when singularities appear, while the quantum protocols allows for efficient recovery of physically relevant quantities without full state tomography. These quantities included pointwise evaluation of the viscosity solution, its gradient, the global minimum and evaluation of functions at the minimiser of the solutions. It is the goal of future work to apply these methods to more specific applications, from front propagation to  optimal control and reinforcement learning.

\section*{Acknowledgements}

SJ thanks Prof. Diogo A. Gomes for pointing to the critical  reference \cite{gomes2007entropy}.  The authors also thank Profs. Remi Abgrall and Yann Brenier  for interesting discussions about Hamilton-Jacobi equations. 

The authors acknowledge the support of the NSFC grant No. 12341104, the Shanghai Pilot Program for Basic Research,  the Science and Technology Commission of Shanghai Municipality (STCSM) grant no. 24LZ1401200 (21JC1402900), the Shanghai Jiao Tong University 2030 Initiative, and the Fundamental Research Funds for the Central Universities. NL is also supported by grant NSFC No. 12471411. 

\bibliography{HJRef}
\appendix

\section{Proof of Lemma~\ref{lem:gomeserror}}
\label{app:background}

    We first show that  the entropy penalty  solution $S^n_\nu$ satisfies a discrete-time Hamilton-Jacobi equation with an error estimate.

\begin{lemma}\label{prop-24}
Let ${S_\nu} \in C^3(\mathbb{R}^d)$, and $K(v)=|v|^2$. Suppose that $h\|\nabla^2 {S_\nu}\|_{L^\infty}(\mathbb{R}^d) $ is smaller than a suitable constant.
Define
\begin{align}\label{S-def}
S^n_\nu (x)=-2\nu \ln {u}^n(x).
\end{align}
Then
\begin{align}\label{error-1}
\frac{S_\nu^{n+1}-S_\nu^n}{h}+H(\nabla S_\nu^n, x)=2a\nu \Delta S_n + O(\mathcal{E})
\end{align}
where
\[
a= \frac{\int_{\mathbb{R}^d}e^{-|w|^2}|w|^2 \, dw }{2\int_{\mathbb{R}^d} e^{-|w|^2} dw}
\]
and the error 
\begin{align}\label{error-2}
\mathcal{E} =O\left( h^2d^3/\nu+ hd^2+(h/\nu)^{3/2} d^2\right) .
\end{align}
%with $C_*$  a suitable positive constant depending  on $d^3 \|\nabla^j S_\nu\|_{L^\infty}$ for $1\le j\le 3$, which implies that 
%\begin{align}\label{error-22}
%\mathcal{E} \sim d^3(h/\nu^2+  (h/\nu)^{1/2}).
%\end{align}
\end{lemma}

\begin{proof}
\eqref{error-2} is the result by combining Proposition 24 and the remark below it in \cite{gomes2007entropy}. Here we need to track the dependence on $d$, and also the dependence of $\nabla^j S_\nu^n$ on $\nu$. We use the estimate that   $\nabla^j S^n_\nu=O(1/\nu^{j-1})$ for $j \ge 1$. 
\end{proof}

%\begin{remark} 
%Since $\|\nabla^2 S_\nu\|_{L^\infty}=O(d/\nu)$, hence $h\|\nabla^2 {S_\nu}\|_{L^\infty}=O(hd/\nu)$. The condition $h\|\nabla^2 {S_\nu}\|_{L^\infty} $ is smaller than a suitable constant implies that $h=o(\nu/d)$. This means that the time step $h$ should be much smaller than the small viscosity coefficient, namely one has to resolve numerically the small viscous layer. This is standard in numerical practice.
%\end{remark}

%specific form of $L(x,v)$. Consider the case of $L(x,v)=|v|^2$. Since
%\[
%\tilde{u}^n(x)=	\int e^{-h|v|^2/(2\nu)}dv\, u^n(x)
%\]
%Hence 
%\[
%\nabla^j \tilde{u}^n(x)=	\int e^{-hL(x,v)/(2\nu)}dv\, %\nabla^j u^n(x)
%\]
%Note $\int e^{-hL(x,v)/(2\nu)}dv=O((\nu/h)^{1/2})$ and $u^n(x)$ is the solution of a linear heat-like equation (see \eqref{eq:heatcontinuous} in Lemma \ref{lem:generalhj}),  thus satisfies maximum principle, hence $C_*=O((\nu/h)^{1/2})$.
%For general $L$, using chain rule to estimate $\nabla^j \int e^{-hL(x,v)/(2\nu)}dv}$ one can get $C_*=O((\nu/h)^{1/2}(h/\nu)^j)$. Since $j\le 3$ thus for the worst case senario 
%$C^*=O((h/\nu)^{5/2}$. 

Next we give the error estimate between the entropy penalisation solution $S^n_{\nu}$ and the solution $S_\nu(t^n, x)$ of the viscous Hamilton-Jacob equation \eqref{eq:hjviscosity0} in the Lemma~\ref{lem:gomeserror}.

\begin{lemma}\label{lemma-20}
Under the same assumption as Lemma \ref{prop-24}, we have
\begin{align}\label{error-100}
S_\nu^n(x)-S_\nu(t^n, x) = O\left(d^3(h/\nu^2+  (h/\nu)^{1/2})\right).
\end{align}
\end{lemma}
\begin{proof}
Since $\partial_t S_\nu=O(1)$, now the result follows easily by comparing \eqref{error-1} and \eqref{eq:hjviscosity0}.
\end{proof}

\section{First-order quantum simulation algorithms with sparse-access} \label{app:simplercostly}

While certainly the cost in the preparation of the desired quantum state can always be optimised (as easily derived from Lemma~\ref{lem:sim}), here for simplicity we provide the cost in simulation using the simplest possible scheme with the easiest-to-prepare ancilla states and simpler Hamiltonian simulation protocols. This is  illustrating that the cost is not prohibitively large even without trying to optimise. 

\begin{lemma} \label{lem:heat1}
    Given the linear heat equation in Eq.~\eqref{eq:heat1} and using $\|\mathbf{u}(0)\|=1$, the cost in the digital quantum simulation of $|u(t)\rangle$ to precision $\epsilon$ is $O(\|\mathbf{u}(t)\|^{-1}d(\nu d/\epsilon+V_{max}/\nu)t/\sqrt{\epsilon})$ using the simplest first-order Schr\"odingerisation scheme. 
\end{lemma}

   \begin{proof} This simulation can be achieved with Schr\"odingerisation, which can achieve optimal sample complexity in the simulation of linear PDEs \cite{optimalschr}. In this case, we are simulating the discrete-variable quantum system, so we first  discretise the spatial variables $x$ and $\eta$ and define $x, \eta$ to lie within a bounded box or interval. Along each dimension, $x$ can be discretised into $N_x$ cells and $\eta$ is discretised into $N_{\eta}$ cells. In this case we can perform Hamiltonian simulation for the Hamiltonian corresponding to Eq.~\eqref{eq:heat1}, which is a matrix of size $N_{\eta}N_x^d \times N_{\eta} N_x^d$:
\begin{align}
    \mathbf{H}(t)=\left(\frac{\nu}{2}\sum_{j=1}^d \hat{P}^2_j+\frac{V(\hat{X},t)}{2\nu}\right)\otimes \hat{D},
\end{align}
which acts on a system of $\ln N_{\eta}+d \ln N_x$ qubits. 
Here $\hat{P}_j$ is a sparse shift matrix and $\hat{X}=\text{diag}(-1/2N_x, \cdots, 1/2N_x)$ for the one-dimensional case, and $\hat{D}=\text{diag}(-1/2N_{\eta}, \cdots, 1/2N_{\eta})$. \\

The optimal cost for the simulation of $\exp(i\bar{\mathbf{H}}t)$ for time-independent Hamiltonian simulation is $\tilde{O}(s_{\bar{H}} \|\bar{\mathbf{H}}\|_{\max}t)$ \cite{berry2015hamiltonian}, where $s_{\bar{H}}$ is the row-sparsity of $\bar{\mathbf{H}}$ and $\|\bar{\mathbf{H}}\|_{\max}$ is the size of the maximum entry in $\bar{\mathbf{H}}$. This can be directly applied in the case when $V(x)$ is time-independent. In this case, $s_{\hat{P}^2}=2d+1$ and $s_{V(\hat{X})}=1=s_{\hat{D}}$, so $s_{\bar{\mathbf{H}}}\leq 2d+2$. Then the max-norm $\|\bar{\mathbf{H}}(t)\|_{max} \leq (\nu/2)\|\sum_j \hat{P}_j^2\|_{max}\|\hat{D}\|_{max}+(V_{max}/\nu)\|\hat{D}\|_{max}$, where $\|\sum_j \hat{P}_j^2\|_{max}\sim O(dN_x^2)$, $V_{max}=\max V(\hat{X})$, $\|\hat{D}\|_{max}\sim O(N_{\eta})$. Thus the cost scales like $O(dN_{\eta}(\nu dN_x^2+V_{max}/\nu)t)$ for the Hamiltonian simulation part.\\

In the case where $V(x,t)$ has explicit dependence on time, time-dependent Hamiltonian simulation problem can be easily transformed into a time-independent Hamiltonian by including an extra clock mode $s$ \cite{timedep}:
\begin{align}
    \bar{\mathbf{H}}=\hat{P}_s \otimes \mathbf{1}+\hat{\mathbf{H}}(\hat{S}), \qquad \mathbf{H}(\hat{S})=\left(\frac{\nu}{2}\mathbf{1}_s \otimes \sum_{j=1}^d \hat{P}_j^2+\frac{1}{\nu}\sum_{l=1}^L f_l(\hat{S})\otimes V_{l}(\hat{X})\right) \otimes \hat{D}.
\end{align}
 The row-sparsity doesn't change from the previous Hamiltonian, and the max-norm only needs to be modified from $V_{max}=\max_x V(x) \rightarrow \max_{t, x} V(t, x)$. Since we are including a new clock mode with the addition of $\hat{P}_s$ and we partition this clock mode into $N_s$ segments, then $\|\hat{P}_s\|_{max} \sim N_s$ \cite{timedep}. Thus the total cost for the simulation of the relevant time-dependent Hamiltonian is $O(N_st+dN_{\eta}(\nu dN_x^2+V_{max}/\nu)t)$. Now, since $N_s$ does not need to exceed $N_x$ or $N_{\eta}$, the first term is dominated by the second term. Finally, to recover the state $|u(t)\rangle$ after the Hamiltonian simulation requires a post-selection process, which has the probability of success that can be amplified to $\sim \|\mathbf{u}(t)\|$. 
 \end{proof}

% Finally, since we have a second order PDE $N_x \sim 1/\epsilon^{1/2}$ where $\epsilon$ is the error in estimating $\mathbf{u}(t)$ from discretising space, and assuming $N_{\eta} \sim N_x$ for simplicity, then the total cost becomes $O(\|\mathbf{u}(t)\|^{-1}d(\nu d/\epsilon+V_{max}/\nu)t/\sqrt{\epsilon})$.
\begin{lemma} \label{lem:quantumheatgeneral}
    Given the linear heat equation in Eq.~\eqref{eq:heatcontinuous}, the cost in the digital quantum simulation of $|u(t)\rangle$ to precision $\epsilon$ is $O(\|\mathbf{u}(t)\|^{-1}d^3 t(|\mu^{max}|+\nu^{max}d/\epsilon)/\sqrt{\epsilon})$, where $\mu^{max}=\max_i \mu_i$, $\nu^{max}=\max_{ij} \nu_{ij}$.
\end{lemma}

\begin{proof}
      As before, to simulate the discrete-variable quantum system, we first  discretise the spatial variables $x$ and $\eta$ and define $x, \eta$ to lie within a bounded box or interval. Along each dimension in$x$, it can be discretised into $N_x$ cells and the $\eta$ domain  is discretised into $N_{\eta}$ cells. In this case we can perform Hamiltonian simulation for the discretised version of the Hamiltonian in Eq.~\eqref{eq:heatcontinuous} which is matrix of size $N_{\eta}N_x^d \times N_{\eta} N_x^d$
\begin{align}
    \mathbf{H}=-\sum_{i=1}^d \mu_i \hat{P}_i \otimes \mathbf{I}+\sum_{j=1}^d \nu_{ij} \hat{P}_i\hat{P}_j\otimes \hat{D},
\end{align}
which acts on a system of $\ln N_{\eta}+d \ln N_x$ qubits. 
Here $\hat{P}_j$ is a sparse shift matrix and $\hat{X}=\text{diag}(-1/2N_x, \cdots, 1/2N_x)$ for the one-dimensional case, and $\hat{D}=\text{diag}(-1/2N_{\eta}, \cdots, 1/2N_{\eta})$. The optimal cost for the simulation of $\exp(i\mathbf{H}t)$ for time-independent Hamiltonian simulation is $O(s_{H} \|\mathbf{H}\|_{\max}t)$, where $s_{H}$ is the row-sparsity of $\mathbf{H}$ and $\|\mathbf{H}\|_{\max}$ is the size of the maximum entry in $\mathbf{H}$. In this case, $s(\sum_i \mu_i \hat{P}_i)=O(d)$, $s(\sum_{ij} \nu_{ij} \hat{P}_i\hat{P}_j \otimes \hat{D})=O(d^2)$. The max-norm $\|\mathbf{H}\|_{max}\leq \|\sum_i \mu_i \hat{P}_i \otimes \mathbf{I}\|_{max}+\|\sum_{ij}\nu_{ij} \hat{P}_i \hat{P}_j \otimes \hat{D}\|_{max}$, where $\|\sum_i \mu_i \hat{P}_i \otimes \mathbf{I}\|_{max} \leq O(d|\mu^{max}|N_x)$, $\|\sum_{ij}\nu_{ij} \hat{P}_i \hat{P}_j \otimes \hat{D}\|_{max} \leq O( \nu^{max}d^2N_x^2 N_{\eta})$, with $\mu^{max}=\max_i \mu_i$, $\nu^{max}=\max_{ij} \nu_{ij}$. Thus, the optimal cost scales as $O(d^3 N_{x}t(|\mu^{max}|+\nu^{max}dN_x N_{\eta}))$. Finally, to recover the state $|u(t)\rangle$ after the Hamiltonian simulation requires a post-selection process, which has the probability of success that can be amplified to $\sim \|\mathbf{u}(t)\|$.\\

In our case here, we do not consider explicit time-dependence since, for the validity of Lemma~\ref{lem:generalhj} to extract $S(x,t)$, we used the fact that we do not have explicit time-dependence in the corresponding Lagrangian. \\

Since we typically use a second order approximation (for example the center difference) for the spatial derivatives, $N_x \sim 1/\epsilon^{1/2}$  where $\epsilon$ is the error in estimating $\mathbf{u}(t)$ from discretising space, and assuming $N_{\eta} \sim N_x$ for simplicity, then the total cost becomes $O(\|\mathbf{u}(t)\|^{-1}d^3 t(|\mu^{max}|+\nu^{max}d/\epsilon)/\sqrt{\epsilon})$.
\end{proof}
\section{Details of proofs in Section~\ref{sec:observables}}

\subsection{Proof details of Lemma~\ref{lem:cvsmin}} \label{app:cvsmin}
Here we show a proof for Eq.~\eqref{eq:smaxapprox} needed in  Lemma~\ref{lem:cvsmin}.  
To compute the integral $I_{\mathcal{N}}=\|\mathbf{u}(t)\|^2=\int e^{S_{\nu}(t, x)/\nu} dx$ we first divide $x \in \mathbb{R}^d$ into the following three regions. 
\begin{enumerate}
    \item  $x \in \mathcal{B}_{\delta}$ is called near-field when $\|x^*_{\nu}-x\|\leq \delta$ is small, and $\mathcal{B}_{\delta}$ is a ball centered at $x^*_{\nu}$ with radius $\delta$. In this region, $S_{\nu}(t, x)-S_{\nu}(t, x^2) \geq c_1|\|x-x^*_{\nu}\|^2$, where $c_1 \propto \lambda_{max}(H_S(x^*_{\nu}))$, which denotes the maximum eigenvalue of $H_S(x^*_{\nu})$. 
    \item  $x \in \mathcal{B}_{R} \backslash \mathcal{B}_{\delta}$ is said to be in the intermediate region when $R > \delta$. Here we can define $\eta=\min_x S_{\nu}(t, x)-S_{\nu}(t, x^*_{\nu})$
    \item  $x \in \mathbb{R}^{d} \backslash \mathcal{B}_R$ is called far-field when $\|x-x^*_{\nu}\| \geq R$.  In this case, we have $S_{\nu}(t, x)-S_{\nu}(t, x^*_{\nu})<c_2\|x^*-x\|^2$ for some $c_2$. Since $\nu \ll 1$, we can apply Laplace's method and show each integral $I_f$ and $I_{\mathcal{N}}$ is dominated by contributions from near-field $x$. 
    \end{enumerate} 
    
    \noindent (1) Near-field $x$:  In the near field we have
\begin{align}
    S_{\nu}(t,x)=S(t,x^*_{\nu})+\frac{1}{2}(x-x^*_{\nu})^T H_S (x-x^*_{\nu})+O(\|x-x^*_{\nu}\|^3),
\end{align}
where $-(x-x^*_{\nu})^T H_S (x-x^*_{\nu}) \leq |\lambda_{max}(H_S(x^*_{\nu}))|\|x-x^*_{\nu}\|^2$, so $S_{\nu}(t,x)-S_{\nu}(t,x^*_{\nu}) \gtrsim c_1\|x-x^*_{\nu}\|^2$ for $c_1 \propto |\lambda_{max}(H_S(x^*_{\nu}))|$. Then rewriting $y=(x-x^*_{\nu})/\sqrt{\nu}$,  
 \begin{align}
      I_{\mathcal{N}, \delta}=\int_{\mathcal{B}_{\delta}} e^{-S_{\nu}(t, x)/\nu}dx\sim\nu^{d/2}e^{-S_{\nu}(t, x^*_{\nu})/\nu}\int_{\mathcal{B}_{\delta}}e^{-y^TH_Sy/2}(1+O(\nu^{3/2}\|y\|^3)dy\sim e^{-S_{\nu}(t, x^*_{\nu})/\nu}\sqrt{\left(\frac{2\pi \nu}{\det H_S}\right)^d} ,
\end{align}
where in the last line the integral over $\mathcal{B}_{\delta}$ can be approximated by the integral over $\mathbb{R}^d$, when $\|y\|\leq \delta/\sqrt{\nu}$ and $\delta/\sqrt{\nu} \gg \sqrt{2/\det H_S}$, where $\sqrt{2/\det H_S}$ is the standard deviation of the Gaussian in $\exp(y^TH_S y/2)$. \\

\noindent (2) Intermediate field $x$: In this limit, we can define $\eta=\min_x S_{\nu}(t, x)-S_{\nu}(t, x^*_{\nu})$. Thus
\begin{align}
    I_{\mathcal{N},R\backslash \delta}=\int_{\mathcal{B}_R \backslash \mathcal{B}_{\delta}}e^{-S_{\nu}(t, x)/\nu} dx= e^{-S_{\nu}(t,x^*_{\nu})/\nu}\int_{\mathcal{B}_R \backslash \mathcal{B}_{\delta}} e^{-(S_{\nu}(t,x)-S_{\nu}(t, x^*_{\nu})/\nu}dx\leq e^{-S_{\nu}(t, x^*_{\nu})/\nu} \text{Vol}(\mathcal{B}_R \backslash \mathcal{B}_{\delta}) e^{-\eta/\nu}.
\end{align}
Due to the exponentially decaying term $\exp(-\eta/\nu)$, this contribution to the intermediate term is exponentially smaller than $I_{\delta}$. For example, we can approximate the scaling $\text{Vol}(\mathcal{B}_R \backslash \mathcal{B}_{\delta}) \sim v^d$ for some $v$, so requiring $\text{Vol}(\mathcal{B}_R \backslash \mathcal{B}_{\delta}) \exp(-\eta/\nu) \ll 1$ just needs $\nu \lesssim \eta/d$. \\

(iii) Far-field $x$: Since the normalisation $\|\mathbf{u}\|^2=\int e^{-S_{\nu}(t,x)/\nu}dx<\infty$, which means that in the far-field limit $u(t)$ must be a decaying function as $x \rightarrow \infty$. In the Laplace method a quadratic decay is often assumed (greater than quadratic decay still satisfies the same bound) as it is the sharpest case that matches the local quadratic expansion about $x^*_{\nu}$, so we can write $S_{\nu}(t,x)-S_{\nu}(t, x^*_{\nu})\sim c_2\|x^*_{\nu}-x\|^2$ for some $c_2$. Then 
\begin{align} \label{eq:farfieldN}
    I_{\mathcal{N}, R}=\int_{\mathbb{R}^d \backslash \mathcal{B}_R}e^{-S_{\nu}(t, x)/\nu}dx= e^{-S_{\nu}(t, x^*_{\nu})/\nu} \int_{\mathbb{R}^d \backslash \mathcal{B}_R}e^{-(S_{\nu}(t, x)-S_{\nu}(t, x^*_{\nu}))/\nu}dx \sim e^{-S_{\nu}(t, x^*_{\nu})/\nu}\int_{\mathbb{R}^d \backslash \mathcal{B}_R} e^{-c_2(x^*_{\nu}-x)^2/\nu} dx. 
\end{align}
Let $z=x-x^*_{\nu}$, then $dx=dz=\omega_{d-1}r^{d-1}dr$, $\omega_{d-1}=2\pi^{d/2}/\Gamma(d/2)$ is the surface area of the unit sphere in of $d-1$ dimensions. So we can rewrite the last integral as 
\begin{align}
   & \int_{\mathbb{R}^d \backslash \mathcal{B}_R} e^{-c_2(x^*_{\nu}-x)^2/\nu} dx=\int_{\|x\|>R}e^{-c_2\|z\|^2/\nu}dz=\omega^{d-1}\int_R^{\infty}e^{-c_2r^2/\nu}r^{d-1}dr \nonumber \\
   &=\frac{\omega^{d-1}\nu^{d/2}}{2c_2^{d/2}}\int_{c_2R^2/\nu}^{\infty}s^{d/2-1}e^{-s}ds=\frac{\omega^{d-1}\nu^{d/2}}{2c_2^{d/2}}\Gamma(d/2,c_2R^2/\nu)
\end{align}
where in the last line we used the substitution $s=c_2r^2/\nu$ and the definition of the incomplete upper gamma function $\Gamma(a, b)=\int_b^{\infty}\tau^{a-1}e^{-\tau}d\tau \sim b^{a-1}e^{-b}(1+(a-1)/b+(a-1)(a-2)/b^2+\cdots)$. For any fixed dimension $d$, one can choose $\nu \ll 1$ so we are  in the regime $d \ll c_2R^2/\nu$. In this regime, one can write $\Gamma(d/2, c_2R^2/\nu)\leq (c_2 R^2/\nu)^{d/2-1} e^{-c_2R^2/\nu}(1+(d/2-1) \nu/(c_2R^2)+O(d^2\nu^2/R^4))$. Insert this back into Eq.~\eqref{eq:farfieldN}
\begin{align}
    I_{\mathcal{N}, R}\lesssim e^{-S_{\nu}(t, x^*_{\nu})/\nu}\frac{1}{\Gamma(d/2)}\left(\frac{\pi \nu}{c_2}\right)^{d/2} \left(\frac{c_2 R^2}{\nu}\right)^{d/2-1}e^{-c_2 R^2/\nu}
\end{align}
where clearly we can see that the exponential fall-off dominates, so $(c_2 R^2/\nu)^{d/2-1}e^{-c_2 R^2/\nu} \ll 1$. Thus this term is also clearly much smaller than the dominant near-field term. \\

After taking the logarithm of the integral, the intermediate and far-field terms will be even more negligible so long as $\nu$ can be made small enough, so the dominant contribution comes from the near-term values and we can write 
\begin{align} \label{eq:normerror}
    |-\nu \ln \|\mathbf{u}(t)\|^2-S_{\nu, \min}|\lesssim \nu d\ln(2\pi \nu/\det H_S)/2 
\end{align}
and other terms in the upper bound are higher order terms in $\nu$.

\subsection{Proof details of Lemma~\ref{lem:cvfmax}} \label{app:fmin}
To prove Eq.~\eqref{eq:fmaxbound}, let us rewrite
\begin{align}
    \langle u(t)|f(\hat{x})|u(t)\rangle=\frac{I_f}{I_\mathcal{N}}, \qquad I_f=\int f(x) e^{-S_{\nu}(t, x)/\nu}dx, \qquad I_{\mathcal{N}}=\int e^{-S_{\nu}(t, x)/\nu}dx.
\end{align}
For each integral $I_f$ and $I_{\mathcal{N}}$, we can divide it into three regions:
\begin{enumerate}
    \item  $x \in \mathcal{B}_{\delta}$ is called near-field when $\|x^*_{\nu}-x\|\leq \delta$ is small, and $\mathcal{B}_{\delta}$ is a ball centered at $x^*_{\nu}$ with radius $\delta$. In this region, $S_{\nu}(t, x)-S_{\nu}(t, x^*_{\nu}) \geq c_1|\|x-x^*_{\nu}\|^2$, where $c_1 \propto \lambda_{max}(H_S(x^*_{\nu}))$, which denotes the maximum eigenvalue of $H_S(x^*_{\nu})$. 
    \item  $x \in \mathcal{B}_{R} \backslash \mathcal{B}_{\delta}$ is said to be in the intermediate region when $R > \delta$. Here we can define $\eta=\min_x S_{\nu}(t, x)-S_{\nu}(t, x^*_{\nu})$
    \item  $x \in \mathbb{R}^{d} \backslash \mathcal{B}_R$ is called far-field when $\|x-x^*_{\nu}\| \geq R$.  In this case, we have $S_{\nu}(t, x)-S_{\nu}(t, x^*_{\nu})<c_2\|x^*_{\nu}-x\|^2$ for some $c_2$ (or some stronger polynomial dependence). Since $\nu \ll 1$, we can apply Laplace's method and show each integral $I_f$ and $I_{\mathcal{N}}$ is dominated by contributions from near-field $x$. 
    \end{enumerate}
\noindent (1) Near-field $x$:  In the near field we have
\begin{align}
    S_{\nu}(t,x)=S(t,x^*_{\nu})+\frac{1}{2}(x-x^*_{\nu})^T H_S (x-x^*_{\nu})+O(\|x-x^*_{\nu}\|^3),
\end{align}
where $(x-x^*_{\nu})^T H_S (x-x^*_{\nu}) \leq |\lambda_{max}(H_S(x^*_{\nu}))|\|x-x^*_{\nu}\|^2$, so $S_{\nu}(t,x)-S_{\nu}(t,x^*_{\nu}) \gtrsim c_1\|x-x^*_{\nu}\|^2$ for $c_1 \propto |\lambda_{max}(H_S(x^*_{\nu}))|$. Then rewriting $y=(x-x^*_{\nu})/\sqrt{\nu}$,  
 \begin{align}
      I_{\mathcal{N}, \delta}=\int_{\mathcal{B}_{\delta}} e^{-S_{\nu}(t, x)/\nu}dx\sim\nu^{d/2}e^{-S_{\nu}(t, x^*_{\nu})/\nu}\int_{\mathcal{B}_{\delta}}e^{-y^TH_Sy/2}(1+O(\nu^{3/2}\|y\|^3)dy\sim e^{-S_{\nu}(t, x^*_{\nu})/\nu}\sqrt{\left(\frac{2\pi \nu}{-\det H_S}\right)^d} ,
\end{align}
where in the last line the integral over $\mathcal{B}_{\delta}$ can be approximated by the integral over $\mathbb{R}^d$, when $\|y\|\leq \delta/\sqrt{\nu}$ and $\delta/\sqrt{\nu} \gg \sqrt{2/\det H_S}$, where $\sqrt{2/\det H_S}$ is the standard deviation of the Gaussian in $\exp(-y^TH_S y/2)$. \\

Similarly for $I_f$, using $f(x)\sim f(x^*_{\nu})+\nabla f(x^*_{\nu})(x-x^*_{\nu})+(x-x^*_{\nu})^T\nabla^2 f(x^*_{\nu})(x-x^*_{\nu})/2$ in the near-field. Then combined with the previous result we get 
\begin{align}    
I_{f, \delta} & \approx \int_{\mathcal{B}_{\delta}}f(x) e^{-S_{\nu}(t, x)/\nu} dx \sim \nu^{d/2}e^{-S_{\nu}(t, x^*_{\nu})/\nu} \int_{\mathcal{B}_{\delta}} (f(x^*_{\nu})+\sqrt{\nu}\nabla f(x^*_{\nu})y+\frac{\nu}{2} y^T \nabla^2 f(x^*_{\nu}) y)e^{-y^T H_S y/2}dy \nonumber \\
& \sim \sqrt{\left(\frac{2\pi \nu}{\det H_S}\right)^d}e^{-S_{\nu}(t, x^*_{\nu})/\nu}\left(f(x^*_{\nu})+\frac{\nu}{2} \nabla^2 f(x^*_{\nu}) \text{Tr}(h^{-1}_S(x^*_{\nu}))\right)
\end{align}
up to order $O(\nu^{3/2})$. \\

\noindent (2) Intermediate field $x$: In this limit, we can define $\eta=\min_x S_{\nu}(t, x)-S_{\nu}(t, x^*)$. Thus
\begin{align}
    I_{\mathcal{N},R\backslash \delta}=\int_{\mathcal{B}_R \backslash \mathcal{B}_{\delta}}e^{-S_{\nu}(t, x)/\nu} \leq \text{Vol}(\mathcal{B}_R \backslash \mathcal{B}_{\delta})e^{-S_{\nu}(t, x^*_{\nu})/\nu} e^{-\eta/\nu}.
\end{align}
Due to the exponentially decaying term $\exp(-\eta/\nu)$, this contribution to the intermediate term is exponentially smaller than $I_{\mathcal{N},\delta}$. Similarly the intermediate region contribution to the $I_f$ integral is also much smaller than the near-field contribution, if we assume there is some constant $\eta'=\min_x f(x^*_{\nu})-f(x)$,
\begin{align}
    I_{f, R \backslash \delta}=\int_{\mathcal{B}_R \backslash \mathcal{B}_{\delta}} f(x) e^{S_{\nu}(t, x)/\nu} dx\lesssim \eta' \text{Vol}(\mathcal{B}_R \backslash \mathcal{B}_{\delta})e^{S_{\nu}(t, x^*_{\nu})/\nu} e^{-\eta/\nu}. 
\end{align}
(iii) Far-field $x$: We are given the fact that the normalisation $\|\mathbf{u}\|^2=\int e^{-S_{\nu}(t,x)/\nu}dx<\infty$, which means that in the far-field limit $u(t)$ must be a decaying function as $x \rightarrow \infty$. In the Laplace method a quadratic decay is often assumed (greater than quadratic decay still satisfies the same bound) as it is the sharpest case that matches the local quadratic expansion about $x^*_{\nu}$, so we can write $S_{\nu}(t,x)-S_{\nu}(t, x^*_{\nu}) \sim c_2\|x^*_{\nu}-x\|^2$ for some $c_2$. Then 
\begin{align} \label{eq:farfieldN}
    I_{\mathcal{N}, R}=\int_{\mathbb{R}^d \backslash \mathcal{B}_R}e^{-S_{\nu}(t, x)/\nu}dx \sim e^{-S_{\nu}(t, x^*_{\nu})/\nu}\int_{\mathbb{R}^d \backslash \mathcal{B}_R} e^{-c_2(x^*_{\nu}-x)^2/\nu} dx. 
\end{align}
Let $z=x-x^*_{\nu}$, then $dx=dz=\omega_{d-1}r^{d-1}dr$, $\omega_{d-1}=2\pi^{d/2}/\Gamma(d/2)$ is the surface area of the unit sphere in of $d-1$ dimensions. So we can rewrite the last integral as 
\begin{align}
   & \int_{\mathbb{R}^d \backslash \mathcal{B}_R} e^{-c_2(x^*-x)^2/\nu} dx=\int_{\|x\|>R}e^{-c_2\|z\|^2/\nu}dz=\omega^{d-1}\int_R^{\infty}e^{-c_2r^2/\nu}r^{d-1}dr \nonumber \\
   &=\frac{\omega^{d-1}\nu^{d/2}}{2c_2^{d/2}}\int_{c_2R^2/\nu}^{\infty}s^{d/2-1}e^{-s}ds=\frac{\omega^{d-1}\nu^{d/2}}{2c_2^{d/2}}\Gamma(d/2,c_2R^2/\nu)
\end{align}
where in the last line we used the substitution $s=c_2r^2/\nu$ and the definition of the incomplete upper gamma function $\Gamma(a, b)=\int_b^{\infty}\tau^{a-1}e^{-\tau}d\tau \sim b^{a-1}e^{-b}(1+(a-1)/b+(a-1)(a-2)/b^2+\cdots)$. For any fixed dimension $d$, we can choose $\nu \ll 1$ so we are in the regime $d \ll c_2R^2/\nu$. In this regime, we can write $\Gamma(d/2, c_2R^2/\nu)\leq (c_2 R^2/\nu)^{d/2-1} e^{-c_2R^2/\nu}(1+(d/2-1) \nu/(c_2R^2)+O(d^2\nu^2/R^4))$. We can insert this back into Eq.~\eqref{eq:farfieldN}
\begin{align}
    I_{\mathcal{N}, R}\lesssim e^{-S_{\nu}(t, x^*)/\nu}\frac{1}{\Gamma(d/2)}\left(\frac{\pi \nu}{c_2}\right)^{d/2} \left(\frac{c_2 R^2}{\nu}\right)^{d/2-1}e^{-c_2 R^2/\nu}
\end{align}
where clearly we can see that the exponential fall-off dominates, so $(c_2 R^2/\nu)^{d/2-1}e^{-c_2 R^2/\nu} \ll 1$. Thus this term is also clearly much smaller than the dominant near-field term. \\

Similarly for the numerator term in the far-field, choosing a function $|f(x)|\leq c_3(1+|x|^L)$ for some $c_3$, then this far-field term is also dominated by the exponentially decaying term so long as $d+m \ll c_2 R^2/\nu$
\begin{align}
    I_{f, R}=\int_{\mathbb{R}^d \backslash \mathcal{B}_R}f(x) e^{-S_{\nu}(t, x)/\nu}dx \lesssim e^{-S_{\nu}(t, x^*_{\nu})/\nu} O((R^2/\nu)^{d/2+L-1}e^{-c_2R^2/\nu}). 
\end{align}
Putting all the results together, we observe that since the contributions to the integrals from the intermediate and far-fields are negligible compared to the near-field, then 
\begin{align}
    \langle u(t)|f(\hat{x})|u(t)\rangle=\frac{I_{f, \delta}+I_{f,R\backslash\delta}+I_{f,R}}{I_{\mathcal{N}, \delta}+I_{\mathcal{N},R\backslash\delta}+I_{\mathcal{N},R}} \approx \frac{I_{f, \delta}}{I_{\mathcal{N},\delta}} \approx f(x^*_{\nu})+\frac{\nu}{2}\nabla^2f(x^*_{\nu})|\text{Tr}(\nabla^2S_{\nu}(t, x^*_{\nu})^{-1})|+O(\nu^{3/2}).
\end{align}
Note that the above holds well when $f(x^*_{\nu})$ is large enough compared to the intermediate and far-field contributions. However, in the case where $f(x)$ is very large when $|x-x^*_{\nu}|$ is very large, but is negligibly small in the near-field limit, then the intermediate and far-field terms could dominate. For example, a sufficient condition would be $f(x^*_{\nu})-f(x) \geq 0$ for all $x$. However, since intermediate and far-field terms are exponentially smaller than the near-field term as $\nu$ increases, this means that $f(x)$ in the intermediate and far-field $x$ region would need to be exponentially larger than $f(x)$ in the near-field region to compete with the near-field limit. Thus, in general, unless $f(x)$ is localised very strongly at the intermediate and far-field regions, then the near-field approximation dominates. This means we want $f(x)/f(x^*_{\nu}) \lesssim \exp(C/\nu)$ where $C$ is a constant for all $x$ in the intermediate and far-field regions. 

\subsection{Proof details of Lemma~\ref{lemma:dvfmax}} \label{app:dvfmax}
 For example, using the midpoint rule in approximating an integral, the total error  $|\int g(x) dx -\sum_j g(x_j) (1/N_x)^d| \sim (1/N_x^2) \int \nabla^2g(x)dx$. When $g(x)=f(x)e^{-S_{\nu}/\nu}$, then 
    \begin{align} \label{eq:gintegral}
        \int \nabla^2 f(x) e^{-S_{\nu}/\nu}dx=\frac{1}{\nu}\int \nabla f(x) \nabla S_{\nu} e^{-S_{\nu}/\nu} dx+\frac{1}{\nu^2}\int f(x)(\nabla S_{\nu})^2e^{-S_{\nu}/\nu}dx-\frac{1}{\nu}\int f(x)\nabla^2 S_{\nu}e^{-S_{\nu}/\nu} dx. 
    \end{align}
    Using the Laplace integral method as done previously in Section~\ref{sec:cvfmin}, the integrals should localise at the value of the minimum of $S_{\nu}$, so the first-order derivative terms $\nabla S_{\nu}$ become negligible. Thus, only the terms $\sim 1/\nu$ dominates. Out of those terms, again using the Laplace integral method where only the $x$ near $x^*$ dominate, then in the integrals above we need to evaluate Gaussian integrals, with the standard deviation term proportional to $\nu$, where the Gaussian integrals would obtain terms $\sim \nu^{d/2}$. Thus, in Eq.~\eqref{eq:gintegral}, we would have surviving terms that scale like $\nu^{d/2-1}$. Then
    \begin{align} \label{eq:dvfmax}
        \frac{\int f(x)e^{-S_{\nu}(t, x)/\nu}dx}{\int e^{-S_{\nu}(t,x)/\nu}dx} \approx \frac{(1/N_x^d)\sum_{j=1}^{N_x^d} f(x_j)e^{-S_{\nu}(t, x_j)/\nu}+O(\nu^{d/2-1}/N_x^2)}{(1/N_x^d)\sum_{j=1}^{N_x^d} e^{-S_{\nu}(t,x_j)/\nu}+O(\nu^{d/2-1}/N_x^2)}= \frac{\sum_{j=1}^{N_x^d} f(x_j)e^{-S_{\nu}(t, x_j)/\nu}+O(1/(\nu N_x^2))}{\sum_{j=1}^{N_x^d} e^{-S_{\nu}(t,x_j)/\nu}+O(1/(\nu N_x^2))}, 
    \end{align}
    thus
    \begin{align}
        |\langle f(\hat{X})\rangle-\langle f(\hat{x})\rangle|\lesssim O(1/(\nu N_x^2)).
    \end{align}

\end{document}